\documentclass[prx,floatfix,superscriptaddress,nofootinbib,notitlepage]{revtex4-1}
\pdfoutput=1
\usepackage{graphicx}
\usepackage{amsthm}
\usepackage{thmtools,thm-restate}
\usepackage{mathtools}
\usepackage{amsmath}
\usepackage{amssymb}
\usepackage{bm}
\usepackage{dsfont}
\usepackage{xcolor}
\usepackage{placeins}
\usepackage{svg}
\usepackage[normalem]{ulem}
\usepackage[colorlinks]{hyperref}
\hypersetup{
	pdfstartview={FitH},
	pdfnewwindow=true,
	colorlinks=true,
	linkcolor=blue,
	citecolor=blue,
	filecolor=blue,
	urlcolor=blue}
\usepackage{multirow}
\usepackage{physics}
\usepackage[capitalise]{cleveref}

\newcommand{\thm}[1]{\hyperref[thm:#1]{Theorem~\ref*{thm:#1}}}
\newcommand{\defn}[1]{\hyperref[defn:#1]{Definition~\ref*{defn:#1}}}
\newcommand{\lem}[1]{\hyperref[lem:#1]{Lemma~\ref*{lem:#1}}}
\newcommand{\prop}[1]{\hyperref[prop:#1]{Proposition~\ref*{prop:#1}}}
\newcommand{\fig}[1]{\hyperref[fig:#1]{Figure~\ref*{fig:#1}}}
\newcommand{\tab}[1]{\hyperref[tab:#1]{Table~\ref*{tab:#1}}}
\renewcommand{\sec}[1]{\hyperref[sec:#1]{Section~\ref*{sec:#1}}}
\newcommand{\app}[1]{\hyperref[app:#1]{Appendix~\ref*{app:#1}}}
\newcommand{\cor}[1]{\hyperref[cor:#1]{Corollary~\ref*{cor:#1}}}
\newcommand{\nn}{\nonumber \\}
\newcommand{\append}[1]{\hyperref[append:#1]{Appendix~\ref*{append:#1}}}

\newcommand{\uin}{\ensuremath{\norm{\mathbf{u}_{\mathrm{in}}}}}
\newtheorem{theorem}{Theorem}
\newtheorem{definition}[theorem]{Definition}
\newtheorem{lemma}[theorem]{Lemma}

\newtheorem{corollary}[theorem]{Corollary}
\newtheorem{prob}{Problem}

\newcommand{\uu}{\ensuremath{\mathbf{u}}}

\newcommand{\xx}{\ensuremath{\mathbf{x}}}
\newcommand{\yy}{\ensuremath{\mathbf{y}}}
\newcommand{\bb}{\ensuremath{\mathbf{b}}}
\newcommand{\dt}[1]{\ensuremath{\frac{\dd#1}{\dd t}}}

\DeclareMathOperator{\identity}{\mathds{I}}
\DeclareMathOperator{\reals}{\mathbb{R}}
\DeclareMathOperator{\carlmatrix}{\ensuremath{\mathcal{A}_\mathit{N}}}

\newcommand{\MQ}{\affiliation{
School of Mathematical and Physical Sciences,
Macquarie University, Sydney, NSW 2109, AU} }

\newcommand{\trNSW}{\affiliation{
Quantum for New South Wales, Sydney, NSW 2000, AU}}

\newcommand{\UofT}{\affiliation{Department of Computer Science, University of Toronto, Canada.}}

\newcommand{\vectorInst}{\affiliation{Vector Institute for Artificial Intelligence, Toronto, Canada.}}

\newcommand{\UTS}{\affiliation{
Centre for Quantum Software and Information,
University of Technology Sydney, NSW 2007, AU}} 

\newcommand{\UMD}{\affiliation{
Joint Center for Quantum Information and Computer Science (QuICS), University of Maryland, College Park, Maryland 20742, USA}}

\begin{document}
\title{Further improving quantum algorithms for nonlinear differential equations via higher-order methods and rescaling}
\date{\today}
\author{Pedro C.~S.~Costa}\MQ\trNSW
\author{Philipp Schleich}\MQ\UofT\vectorInst
\author{Mauro E.~S.~Morales}\UTS\UMD
\author{Dominic W.~Berry}\MQ

\begin{abstract}
   The solution of large systems of nonlinear differential equations is needed for many applications in science and engineering. In this study, we present three main improvements to existing quantum algorithms based on the Carleman linearisation technique.    
   First, by using a high-precision technique for the solution of the linearised differential equations, we achieve 
   logarithmic dependence of the complexity on the error 
   and near-linear dependence on time.
   Second, we demonstrate that a rescaling technique can considerably reduce the cost, which would otherwise be exponential in the Carleman order for a system of ODEs, preventing a quantum speedup for PDEs. Third, we provide improved, tighter bounds on the error of Carleman linearisation.
   We apply our results to a class of discretised reaction-diffusion equations using higher-order finite differences for spatial resolution.
   We show that providing a stability criterion independent of the discretisation can conflict with the use of the rescaling due to the difference between the max-norm and 2-norm.
   An efficient solution may still be provided if the number of discretisation points is limited, as is possible when using higher-order discretisations.
 
\end{abstract}

\maketitle

\tableofcontents

\section{Introduction}

Many processes in nature exhibit nonlinear behaviour that is not sufficiently approximated by linear dynamics. Examples range from biological systems, chemical reactions, fluid flow, and population dynamics to problems in climate science. 
Because the Schr\"odinger equation is linear, quantum algorithms are more naturally designed for linear ordinary differential equations (ODEs), as in \cite{berry2014high,berry2017quantum,berry2022quantum,childs2020quantum,an2022theory,fang2023time,krovi2023improved,PhysRevLett.131.150603,an2023quantum}.
These algorithms are normally based on discretising time to encode the ODE in a system of linear equations, then using quantum linear system solvers~\cite{HHL09,costa2022optimal}.
Others are based on a time-marching strategy, solving the ODE using a linear combination of unitary dynamics \cite{PhysRevLett.131.150603,an2023quantum}.
The advantage of these quantum algorithms is that they naturally provide an exponential speedup in the dimension (number of simultaneous equations), similar to the simulation of quantum systems, with the caveat that the solution is encoded in the amplitudes of a quantum state.

The most natural way to approximate quantum solutions of partial differential equations (PDEs) is to first discretise the PDE to construct an ODE, which can then be solved using a quantum ODE algorithm.
Although one might expect an exponential speedup in the number of discretisation points (which would give the dimension for the ODE), this is not realised.
This approach to solve PDEs typically has a more modest polynomial speedup over classical methods due to the norm or condition number of the matrices resulting from the discretisation.
Reference \cite{CJS13} suggested using preconditioners, though later work found that the preconditioners did not significantly reduce the condition number.
Reference~\cite{childs2021high} approached this problem by using higher-order finite difference stencils as well as a pseudo-spectral method.
Alternatively, one can use a wavelet-based preconditioner to achieve scaling independent of the condition number in some cases \cite{bagherimehrab2023fast}.
References~\cite{jin2022quantumSchrodingerisation,jin2022quantumSchrodingerisationDetails} introduce a new method using a variable transformation which provides solutions of PDEs in an equivalent frame using quantum simulation techniques.

Quantum algorithms for nonlinear differential equations were addressed in early work which had very large complexity \cite{leyton2008quantum}.
Later proposals were based on the nonlinear Schr\"odinger equation \cite{lloyd2020quantum}, or an exact mapping of the nonlinear Hamilton-Jacobi PDE into a linear PDE \cite{jin2022quantumObservable, jin2022timeLinear}.
Possibly the most promising approach for the solution of nonlinear ODEs is based on Carleman linearisation~\cite{carleman1932application}, which involves transforming the nonlinear differential equation into a linear differential equation on multiple copies of the vector. 
This approach can be realised particularly easily for differential equations with polynomial nonlinearities and has been applied to quantum algorithms in the case of a quadratic function as the nonlinear part of the ODEs~\cite{liu2020efficient}, for a higher power of the function for a specific PDE~\cite{an2022efficient}, and for the notorious Navier-Stokes equations~\cite{li2023potential}. The homotopy perturbation method to tackle quadratic nonlinear equations in Ref.~\cite{xue2022quantum} leads to similar equations as Carleman linearisation. 

However, most approaches to quantum Carleman linearisation~\cite{liu2020efficient,an2022efficient} applied to PDEs suffer from high error rates due to simple discretisation schemes for the underlying PDE in time and space.
One work \cite{krovi2023improved} does use an improved discretisation in time via a truncated Taylor series \cite{berry2017quantum}.
Using a finer discretisation to achieve a given accuracy results in higher complexity, typically due to the complexity depending on the matrix condition number.
That can result in the complexity being the same or worse than that for classical solution.
Another difficulty in the use of Carleman linearisation in prior work is that the component with the solution may have low probability to be measured.
In this study, we provide three improvements over prior work.
First, we use higher-order methods in the time evolution as well as for the spatial discretisation for PDEs.
Second, we use rescaling in order to eliminate the problem of the low probability of the component with the solution for an intrinsic system of ODEs.
(Reference \cite{krovi2023improved} mentioned a rescaling at one point, though their explanation is unclear and it is unclear if they are using it.)
Third, we provide a tighter bound on the error in Carleman linearisation by explicitly bounding repeated integrals.
In the case of a PDE, the appropriate stability condition is in terms of the max-norm.
However, the interaction of the requirement of the rescaling with the Carleman linearisation and the stability requirement for the ODE solver means that a stronger stability criterion is needed to enable efficient solution.

It is important to note that in the case of PDEs the factor that is exponential in $N$ in prior work \cite{an2022efficient} would give a large power in the number of grid points.
Since a simple classical algorithm would have complexity linear in the number of grid points, the quantum speedup would be eliminated.
Our work demonstrates that quantum computers can provide a sublinear complexity in the number of grid points for nonlinear PDEs, as well as establishing the limitations to this type of approach.
We present an overview of the general solution procedure of nonlinear differential equations on quantum computers in relation to the present work in~\cref{fig:main_flow}.

The paper is organized as follows.
We summarise the contributions and significance of the current work in \cref{sec:contribution}, then specify the problem and our method of solution in \cref{sec:problem-description}.
We then describe the Carleman linearisation techniques in \cref{sec:carleman}, with a summary in \cref{sec:backgound}.
\cref{sec:rescaled_ODEs_Trunc} explains how rescaling the solution vector in Carleman linearisation can boost the amplitude for the solution.
We provide improved error bounds for the rescaled case in~\cref{sec:carleman}.
Subsequently, we apply the ODE solver from Ref.~\cite{berry2022quantum} in \cref{sec:ODEs} to improve over the forward Euler method, providing logarithmic as opposed to linear dependency on the error, and linear as opposed to quadratic dependency on the simulation time.
We also provide complexities adjusted to the proposed rescaling.
We analyse the case of time-independent quantities, though the approach can also be applied for time-dependent equations.
We then show how we can apply our results to obtain the solution of a specific case of nonlinear PDEs in \cref{sec:nonlin-pde-problem}, explaining the  limitations. Finally, we conclude in \cref{sec:conclusion}.

\begin{figure}
    \centering\includegraphics[width=\linewidth]{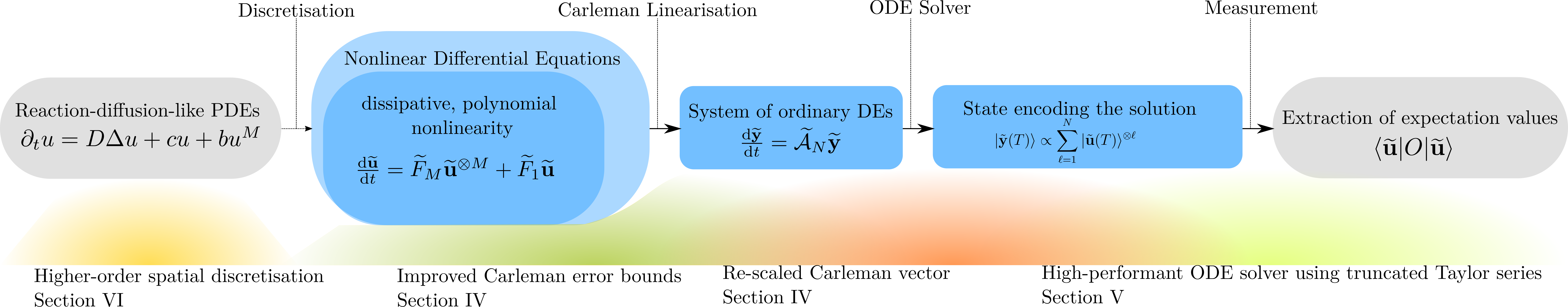}
    \caption{Overview of solution pipeline for nonlinear (partial) differential equations in relation to presented contributions and location in paper. The PDE is approximated by a vectorial ODE via discretisation, $u \longrightarrow \uu$ and we use a scaling $\uu \overset{1/\gamma}{\longrightarrow}\widetilde\uu$ as described in \cref{defn:rescaled-carleman}.}
    \label{fig:main_flow}
\end{figure}

\section{Contribution and significance}\label{sec:contribution}
We give a number of improvements to the solution of nonlinear ODEs and PDEs.
\begin{enumerate}
    \item We use a higher-order method for discretisation of the PDE, which will be required in practice because the stability of the solution will require that the number of points is not too large.
    \item We use rescaling of the components in the Carleman linearisation in order to ensure that the first component containing the solution can be obtained with high probability.
    We show that the amount of rescaling that can be used is closely related to the stability of the equations.
    \item We provide a much tighter analysis of the error due to the Carleman linearisation for ODEs, and extend this analysis to PDEs.
    This analysis is dependent on the stability and the discretisation of the PDE.
\end{enumerate}

All these improvements are dependent on the stability of the equations, which is required for the quantum algorithm to give an efficient solution.
The equations have a linear dissipative term and the nonlinear growth term.
As the input is made larger, the nonlinear term will cause growth and make the solution unstable.
Therefore, for the solution to be dissipative, the input needs to be sufficiently small that the dissipative term dominates.
In the ODE case the input is a vector $\mathbf{u}_{\mathrm{in}}$, and the stability criterion can be given in terms of the 2-norm of that vector.
In the case of the PDE, it is more appropriate to give the stability criterion in terms of the max-norm, because the 2-norm will change depending on the number of discretisation points.

Giving the stability criterion in terms of the max-norm then makes the analysis of higher-order discretisations challenging.
The reason is that, while the first-order discretisation of the Laplace operator is stable in terms of the max-norm, the higher-order discretisations no longer are.
In the analysis of the Carleman linearisation error it is required that the equations are stable.
For the ODE this stability in terms of the 2-norm enables the 2-norm of the error to be bounded.
For the PDE, stability in terms of the max-norm enables the max-norm to be bounded, but the higher-order discretisation complicates the analysis and means slightly stronger stability is required.

The reason why rescaling is needed is that the Carleman method involves constructing a quantum state with a superposition of one copy of the initial vector, two copies, and so forth up to $N$ copies.
If the initial vector is not normalised, then this means that there can be an exponentially large weight on the largest number of copies, whereas the first part of the superposition with a single copy is needed for the solution.
In order to ensure the probability for obtaining that component is not exponentially small, the Carleman vector needs to be rescaled by (at least) the 2-norm so that there is sufficient weight on that first component.
Even if $N$ is small, this feature means that rescaling is essential in order to obtain any speedup over classical algorithms for PDEs.
Without the rescaling, the complexity is superlinear in the number of grid points.

In order to ensure that the same equations are being solved, the components of the matrix need a matching rescaling, which can increase the weight of the nonlinear part (causing growth) as compared to the linear dissipative part.
In the case of an ODE, we show that if the original nonlinear equation is dissipative then the linear ODE obtained from Carleman linearisation is also stable.
That stability is required for the quantum ODE solver to be efficient.
If the ODE is not stable, then the condition number will be exponentially large (in time), which causes the linear equation solver to have exponential complexity.

Similarly for the discretised PDE, there needs to be rescaling by the 2-norm in order to ensure there is adequate weight on the first component of the solution.
The key difference now is that the stability of the equations is given in terms of the max-norm, but the rescaling is by the 2-norm which is typically larger.
That rescaling can give a linear ODE that is no longer stable, which in turn would mean an exponential complexity of the algorithm.
That is perhaps surprising, because the original nonlinear equation is stable.

However, if the PDE is sufficiently dissipative, then the discretised equation will still satisfy the stability criterion in terms of the 2-norm, and there will still be an efficient quantum algorithm.
Because the 2-norm will increase without limit with the number of discretisation points, it is then crucial to minimise the number of discretisation points used.
That further motivates using the higher-order discretisation of the PDE, because that minimises the number of discretisation points.

\section{Problem description and solution strategy}\label{sec:problem-description}

The main focus of this work is the treatment of nonlinear differential equations,  when we have an arbitrary power $M$ in the nonlinear ODE problem on quantum computers, that is
\begin{equation}
\dt{\mathbf{u}} = F_1\mathbf{u} + F_M\mathbf{u}^{\otimes M} \, ,
\end{equation}
followed by its application to the nonlinear reaction-diffusion PDE,
\begin{equation}
   \partial_t  u (\xx,t) = D \Delta u(\xx,t) + c  u(\xx,t)+ b u^{M} (\xx,t). 
\end{equation}
In the following we summarise the problem description and solution strategy for these two equations.

\subsection{The ODE problem}
Here, we present the problem of solving the nonlinear ODE, including the variable definitions and the dissipativity condition needed for an efficient quantum algorithm.
\begin{prob}\label{defn:nonlin-ode-problem} We consider the solution of a system of nonlinear (vectorial) dissipative  ODEs of the form
\begin{equation}
\label{eq:nonlin-ode-problem}
\dt{\mathbf{u}} = F_1\mathbf{u} + F_M\mathbf{u}^{\otimes M} \, ,
\end{equation}
with initial data
\begin{equation}
\mathbf{u}(t=0) = \mathbf{u}_{\mathrm{in}},
\end{equation}
where  $\mathbf{u}=(u_1\cdots,u_{n})^T \in \mathbb{R}^{n}$  with time-dependent components $u_j=u_j(t)$ for $t \in [0,T]$ and $j \in [n]$, using the notation $[n] = \{1, 2, \dots, n\}$. 
The  matrices $F_M \in \mathbb{R}^{n \times n^M}$,  $F_1 \in \mathbb{R}^{n \times n}$ are time-independent.
We denote the eigenvalues of $(F_1+F_1^\dagger)/2$ by $\lambda_j$, and the dissipativity condition means that $\Re(\lambda_j) < 0$.
Denoting the maximum eigenvalue by $\lambda_0$, we require that $R<1$, where
\begin{equation}
\label{eq:R}
{R \coloneqq \frac{\norm{F_M}\cdot \uin^{M-1}}{\abs{\lambda_0}} } .
\end{equation}
The task is to output a state $\ket{\uu}$ encoding the solution to \cref{eq:nonlin-ode-problem} at time $T$.
\end{prob}

To solve \cref{defn:nonlin-ode-problem}, we first map the finite-dimensional system of nonlinear differential equations in \cref{eq:nonlin-ode-problem} to an infinite-dimensional, linear set of ODEs that can be truncated to some order $N$.
This mapping is the Carleman linearisation technique~\cite{carleman1932application}, which has previously been applied to quantum algorithms in Refs.~\cite{liu2020efficient,an2022efficient,krovi2023improved}.
Next, we show that by rescaling the linearised ODEs, we can reduce the complexity of the quantum algorithm.
This is followed by improved error bounds due to Carleman linearisation for the rescaled variable and an estimate of the overall complexity for obtaining the solution of the truncated linearised ODE. 

In contrast to Refs.~\cite{liu2020efficient,krovi2023improved}, we do not consider the driving term; on the other hand, we explore arbitrary nonlinear powers in the ODE problem rather than constrained to the quadratic case as in Refs.~\cite{liu2020efficient,krovi2023improved}. When we have an arbitrary power $M$ in the nonlinear ODE
it is more challenging to include the driving term $F_0$, because $F_0$ will produce characteristics of a more general polynomial of order $M$ as opposed to just a single component. 
Therefore, to analyse the driving term we would also need to consider a general polynomial of order $M$ for the nonlinear part of the ODE problem. We leave that considerably more complicated analysis to future work.

The solution of a linearised form of \cref{defn:nonlin-ode-problem} relies on oracles for $F_1$, $F_M$, and the initial vector.
We show in \cref{sec:ODEs}, that the complexity of the solution in terms of calls to oracles for $F_1$ and $F_M$ scales as
\begin{equation}
    \order{
    \frac{1}{\sqrt{1-R^{2/(M-1)}}}\frac{\uin}{\norm{\uu(T)}} \lambda_{F_1} T N\log({\frac{N}{\varepsilon}}) \log(\frac{ N\lambda_{F_1}T}{\varepsilon})}.
\end{equation}
In this complexity, $\varepsilon$ is the allowable error and $\lambda_{F_1}$ is the $\lambda$-value for block encoding $F_1$ (with an extra assumption on the efficiency of the block encoding of $F_M$).
An important quantity here is the Carleman order $N$, which can be chosen logarithmically in the allowable error provided $R<1$.
For the complexity in terms of calls to preparation of the initial vector, there is an extra factor of $N$, but the final log factor can be omitted, so the overall complexity is similar.
Without the rescaling, there would be an extra factor in the complexity $\order{\uin^N}$ that is exponential in $N$.
Even though $N$ can be chosen logarithmic in the other parameters, that would still result in large complexity.

The result as given in Ref.~\cite{an2022efficient} has that problem.
The complexity from Ref.~\cite{an2022efficient} is (using Eq.~(4.2) of that work and replacing $a$ in their notation with $c$ in our notation)
\begin{equation}
\label{eq:An_complex}
    \order{   
    \frac{1}{G^2\varepsilon} sT^2D^2d^2n^{4/d}N^3\norm{\mathbf{u}_{\text{in}}}^{2N} \text{poly}\left(\log(\frac{cDdMn^{1/d}NsT}{G\varepsilon})\right)
    } ,
\end{equation}
where $G$ denotes the average $\ell_2$ solution norm of the history state, and $s$ is the maximum sparsity of $F_1$, $F_M$.
The factor $\uin^{2N}$ exponential in $N$ is due to the higher-order components of the Carleman vector without rescaling.
They also have a factor of $T^2$ rather than $T$, which is due to using a simple forward Euler scheme in time.
We also give a further improvement in the polynomial factor of $N$, with our scaling being $N$ in comparison to their $N^3$.

\subsection{Carleman solver for the reaction-diffusion equation}

A large system of ODEs of the form in \cref{eq:nonlin-ode-problem} may arise from discretisation of partial differential equations. Specifically, we can derive the nonlinear differential equation resulting from the discretisation of a nonlinear reaction-diffusion PDE similar to Ref.~\cite{an2022efficient},
\begin{equation}
\label{eq:PDE1}
   \partial_t  u (\xx,t) = D \Delta u(\xx,t) + c  u(\xx,t)+ b u^{M} (\xx,t). 
\end{equation}
This equation will be stable according to a criterion that depends on the max-norm of $u(\xx,t)$, in contrast to the condition for the ODE that is based on the 2-norm.
Discretising this PDE into an ODE, the stability condition $R<1$ would be stronger and depend on the number of discretisation points.
That condition is stronger than necessary for the PDE, but after we use Carleman linearisation to give a linear ODE it requires $R<1$ for stability.
This means that the stability condition needed for the quantum algorithm is stronger than that for the original PDE.

We explore techniques of finite-difference methods with higher-order approximations for the spatial discretisation of the PDEs. Our improved nonlinear ODE solver is then applied to the reaction-diffusion equation \cref{eq:PDE1}, with $F_1$ resulting from the Laplacian discretisation and $F_M$ giving the non-linearity from the PDE. 
The overall procedure is illustrated in~\cref{fig:main_flow}. 

We then show in \cref{cor:pde-complexity} that for this PDE, the overall cost for the solution in terms of calls to the oracles that block encode $F_1$ and $F_M$ is 
\begin{equation}\label{eq:pdecomplexity-in-section-2}
\order{
    \frac{1}{\sqrt{1-R^{2/(M-1)}}}\frac{\uin}{\norm{\uu(T)}} (dDn^{2/d} + \abs{c}) T N\log({\frac{N}{\varepsilon}}) \log(\frac{ N(dDn^{2/d} + \abs{c})T}{\varepsilon})}, 
\end{equation}
where we have used $n$ gridpoints in total for the spatial discretisation of the $d$-dimensional PDE given in \cref{eq:PDE1}.

Classically, it is less useful to perform linearisation by the Carleman procedure, because the system size grows exponentially with the truncation number $N$ making the simulation prohibitively costly.
In general, explicit time-stepping methods like forward Euler or Runge-Kutta schemes do not rely on linearisation of the underlying differential equations. However, (semi-)implicit schemes which exhibit more favourable numerical stability rely on inversion of the system. This either requires linearisation (e.g., Carleman or Koopman-von-Neumann schemes) or methods to solve nonlinear systems, such as Newton-Raphson, which rely on a good initial guess and require inversion of a Jacobian matrix.

\section{Quantum Carleman solver with rescaling and improved error bounds on Carleman truncation}\label{sec:carleman}

\subsection{Background on Carleman linearisation}
\label{sec:backgound}

We start with the Carleman linearisation for the initial value problem described by the $n$-dimensional equation with a nonlinearity of order $M$ as given in \cref{eq:nonlin-ode-problem}.  We recall the dissipativity assumption on $F_1$, i.e., the eigenvalues of $(F_1+F_1^\dagger)/2$ are purely negative. The quantity $\lambda_0$, the eigenvalue closest to zero, thus gives the weakest amount of dissipation.  
This way, $R$ in \cref{eq:R} can be used to quantify the strength of the nonlinearity of the problem. As shown in Ref.~\cite{liu2020efficient}, there exists a quantum algorithm that can solve \cref{eq:nonlin-ode-problem} efficiently whenever $R<1$. Furthermore, for $R\geq \sqrt{2}$, the problem was shown to be intractable on quantum computers.

Next, we briefly outline the key idea of the Carleman linearisation. First, notice that   
\begin{equation}
\label{eq:tensM}
\mathbf{u}^{\otimes M}= (u_1^M,u_1^{M-1}u_2,\dots,u_1u_n^{M-1}, u_2u_1^{M-1},\dots,u_n^{M-1} u_{n-1},u_n^M)^T \in \mathbb{R}^{n^M}.
\end{equation} 
In particular, for $M=2$, the Kronecker product gives
\begin{equation}
\mathbf{u}^{\otimes2}= (u_1^2,u_1u_2,\dots,u_1u_n,u_2u_1,\dots,u_nu_{n-1},u_n^2)^T \in \mathbb{R}^{n^2}.
\end{equation} 
Now, define a new variable consisting of Kronecker powers of the solution vector
\begin{equation}\label{eq:carleman-vector}
\yy_1=\uu,\; \yy_2=\uu^{\otimes 2},\;\dots,\;\yy_N=\uu^{\otimes N},\dots, 
\end{equation}
which we can summarise as a vector $\yy = [\yy_1, \yy_2, \dots, \yy_N, \ldots]^T$.
If we consider the time-derivative, we can identify the time-independent matrices $F_M \in \mathbb{R}^{n \times n^M}$ and  $F_1 \in \mathbb{R}^{n \times n}$ as follows,
\begin{align}
\label{eq:Carl_line}
\dt{\mathbf{y}_j} &= \dt{\mathbf{u}^{\otimes j}} \nn 
&= \dt{\mathbf{u}}\otimes \mathbf{u} \otimes \cdots \otimes \mathbf{u} + \dots + \mathbf{u}\otimes \mathbf{u} \otimes \dots \otimes \dt{\mathbf{u}} \nn
&= (F_M \mathbf{u}^{\otimes M}) \otimes \mathbf{u} \otimes \dots \otimes \mathbf{u} + \dots + \mathbf{u}\otimes \mathbf{u} \otimes\cdots \otimes  (F_M \mathbf{u}^{\otimes M}) \nn 
&\quad+ (F_1\mathbf{u}) \otimes \mathbf{u} \otimes\cdots \otimes \mathbf{u} + \dots + \mathbf{u}\otimes \mathbf{u} \otimes\dots \otimes  (F_1  \mathbf{u}).
\end{align}
We can write this in compact form,
\begin{equation}
 \dt{\mathbf{y}_j} =   A_{j+ M - 1}^{(M)}\mathbf{y}_{j + M - 1} + A_{j}^{(1)}\mathbf{y}_{j},
\end{equation}
where $A^{(M)}_{j+M-1} \in \mathbb{R}^{n^j \times n^{j + M-1}}$ and $A^{(1)}_{j} \in \mathbb{R}^{n^j \times n^{j}}$  with
\begin{align}
\label{eq:carleman-assembly}
A^{(M)}_{j+M-1} &= F_M\otimes \identity^{\otimes (j-1)} + \identity\otimes F_M\otimes \identity^{\otimes(j-2)} + \dots + \identity^{\otimes (j-1)}\otimes F_M  \nonumber\\
 A^{(1)}_{j} &= F_1\otimes \identity^{\otimes (j-1)} + \identity\otimes F_1\otimes \identity^{\otimes(j-2)} + \dots + \identity^{\otimes (j-1)}\otimes F_1,
\end{align}
where the $\identity$ operation is the identity with the same domain as $F_1$, i.e., $\reals^{n\times n}$.

This results in an infinite-dimensional linear system, as there is no bound on the range of $j$. To make this computationally feasible, we restrict to $j\in[N]$ for some $\mathbb{N}\ni N > M$. 
Further, we can see that $N> M$ is a requirement in order to be able to capture any effects coming from a nonlinearity of order $M$.
This allows one to write down a matrix form,
\begin{equation}
\label{eq:CarlODE}
\dt{\yy} = \mathcal{A}_N\mathbf{y}, 
\end{equation}
with
\begin{equation}
\label{eq:carlmatrix}
\mathcal{A}_N=
\begin{bmatrix}
A_1^{(1)} & 0          & \cdots & 0        & A^{(M)}_{M}& 0             & \cdots       & 0          \\
0         & A_2^{(1)}  & \cdots & 0        & 0          & A^{(M)}_{M+1} & 0            & \vdots     \\
\vdots    & 0          & \ddots &          &            & 0             & \ddots       & 0          \\
          &            & \ddots & \ddots   &            &               & \ddots       & A_N^{(M)}  \\
          &            &        &          & \ddots     &               &              & 0          \\
          &            &        &          & \ddots     & \ddots        &              & \vdots     \\
\vdots    &            &        &          &            & 0             & A_{N-1}^{(1)}& 0          \\
0         & 0          & \cdots &          &            & \cdots        & 0            & A_N^{(1)}  \\
\end{bmatrix}.
\end{equation}
The matrix $\mathcal{A}_N\in \mathbb{R}^{N_{\mathrm{tot}}\times N_{\mathrm{tot}}}$ is called the  Carleman matrix with truncation order $N$, where ${N_{\mathrm{tot}}=\sum_{j=1}^N n^j} = \frac{n(n^N-1)}{n-1}$. The non-truncated, infinitely large matrix we call $\mathcal{A}$.
As the dimensionality of the system is exponential in the order of Carleman truncation (see \cref{fig:carlmatrix-snippet}), this technique tends to be intractable for practical applications on classical computers. 
\begin{figure}
    \centering
\includegraphics[width=.4\linewidth]{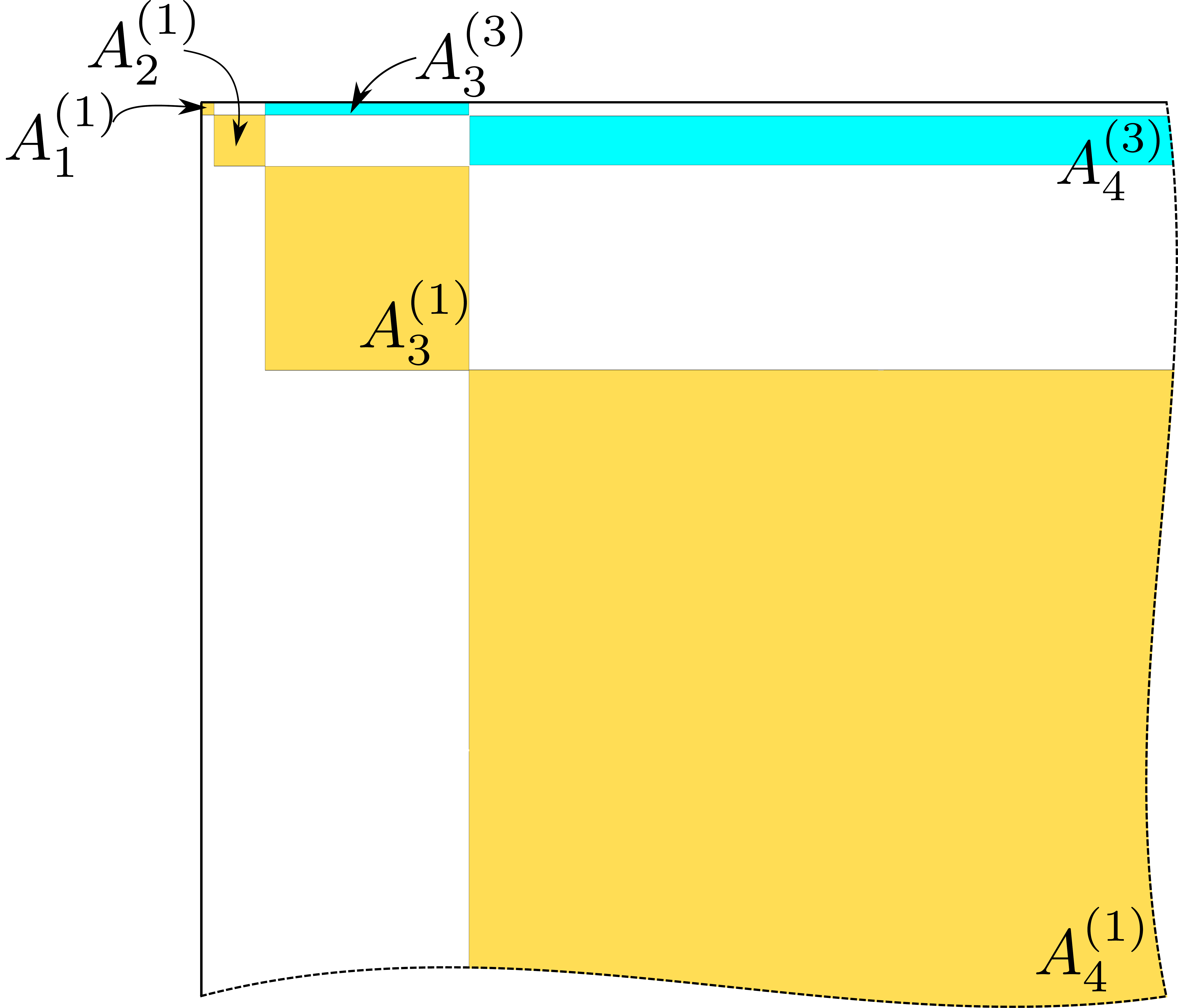}
    \caption{Depiction of a snippet of $\carlmatrix$ for $M=3$ until $N=4$. Given the exponential increase in size, only a fraction of $N=4$ is shown. The diagonal blocks correspond to linear terms of the ODE, the upper-diagonal blocks to a nonlinearity on the $(M-1)$st off-diagonal.}
    \label{fig:carlmatrix-snippet}
\end{figure}

The simple block structure of the matrix $\mathcal{A}_N$ enables us to obtain the upper bound for $\|\mathcal{A}_N\|$ in terms of the norms of the submatrix of $\mathcal{A}_N$, that is 
\begin{align}
\label{eq:norm}
\|\mathcal{A}_N\| &\leq \max_{1\le j\le N} \|A^{(1)}_{j}\| + \max_{1 \le j \le N+M-1} \|A^{(M)}_{j}\| \nn
&= N\|F_1\| + (N-M+1)\|F_M\| \, .
\end{align}
A similar relation holds for the $\lambda$-values, which is important for the estimation of the complexity of our quantum algorithm. In what follows, we present a lemma that allows us to quantify the total error involved in the Carleman truncation. Our lemma considers the error from the Carleman linearisation for the rescaled nonlinear ODE problem when we have an arbitrary power $M$ for the function, as opposed to the quadratic case without the rescaling given in \cite{liu2020efficient}.
To that end, we will first present said rescaling.

\subsection{A rescaled Carleman solver}\label{sec:rescaled_ODEs_Trunc}
We will motivate this rescaling by looking at the measurement probabilities of components in the vector $\yy = [\uu, \uu^{\otimes2}, \ldots, \uu^{\otimes N}]^T$. Recall that the sole entry we are interested in measuring will be $\yy_1\equiv \uu$.
The standard way to encode the solution $\uu(t)$ in a computational basis $\{\ket{j}\}$ is
\begin{equation}
\label{eq:targ_state}
\ket{\mathbf{u}(t)}=\sum_{j=1}^{n}u_j(t)\ket{j}.   
\end{equation}
Analogously, components $\ket{\yy_m}$ of $\yy$ are written as a quantum state as 
\begin{align}
\ket{\mathbf{y}_m(t)}& =\sum^n_{j_1,\dots,j_m  = 1}u_{j_1}(t)\cdots u_{j_m}(t)\ket{m,j_1 \cdots j_m , 0^{\otimes \log(n)(N-m)}}\nn
&=\sum^n_{j_1,\dots,j_m  = 1}u_{j_1}(t)\cdots u_{j_m}(t)\ket{y_m^{(j_1j_2\cdots j_m)}},
\end{align}
with
\begin{equation}
\ket{y_m^{(j_1j_2\cdots j_m)}} \coloneqq \ket{m,j_1 \cdots j_m , 0^{\otimes \log(n)(N-m)}}.    
\end{equation}
This follows the state encoding outlined in Appendix~3.C in Ref.~\cite{liu2020efficient}, where in each step up to the largest order $N$, extra dimensions are padded in the form of $\ket{0}$'s to avoid the structure of a superposition over components of different size. The first register is set to $m$ so we can distinguish the order by measurement of a subsystem.
Then, we can write the full vector $\ket{\yy(t)}$ as follows:
\begin{align}
\label{eq:structure-final-state}
\ket{\mathbf{y}(t)} &= \sum_{j_1=1}^{n} u_{j_1}(t)\ket{y_1^{(j_1)}} + \sum_{j_1,j_2=1}^{n}  u_{j_1}(t) u_{j_2}(t)\ket{y_2^{(j_1j_2)}} + \cdots +\sum_{j_1,\dots,j_N=1}^{n} u_{j_1}(t)\dots \bar{u}_{j_N}(t)\ket{y_N^{(j_1j_2\cdots j_N)}}.
\end{align}
For a normalised quantum state, the amplitudes $u_{j_l}(t)$ in \cref{eq:structure-final-state} need to be normalised so that $\braket{\yy}=1$.
We then have to consider the normalisation factor $1/\sqrt{V_N}$ where
\begin{equation}\label{eq:normalis}
    V_N=\norm{\yy}^2 = \sum_{\ell=1}^N \|\uu(t)\|^{2\ell} = \|\uu(t)\|^2 \frac{1-\|\uu(t)\|^{2N}}{1-\|\uu(t)\|^2} \, .
\end{equation}
Note that this formula does not work in the case that $\|\uu(t)\|=1$.
We therefore adopt the convention that wherever there appears a ratio of this form, for $\|\uu(t)\|=1$ it takes the value in the limit $\|\uu(t)\|\to 1$, so
\begin{equation}
    \|\uu(t)\|^2 \frac{1-\|\uu(t)\|^{2N}}{1-\|\uu(t)\|^2} \to N \, .
\end{equation}

The solution of the nonlinear ODE is given by the first component, where the probability is given by
\begin{equation}
 {P}(\yy_1(t)) = \sum_{j_1=1}^n\left|\braket{y_1^{(j_1)}}{\yy(t)}\right|^2 = \frac{1}{V_N} \sum^n_{j_1=1}|{u}_{j_1}(t)|^2=\frac{1-\norm{\uu(t)}^2}{1-\norm{\uu(t)}^{2N}} .
\end{equation}
From this equation, we see that as we increase the Carleman truncation order we also increase $V_N$, which suppresses the probability of extracting the desired component. This brings an exponential cost in $N$ for the algorithm due to the $\order{1/\sqrt{{P}(\yy_1(t))}}$ rounds of amplitude amplification needed at the end. To avoid this high cost in the algorithm, we propose the following rescaling, which can significantly reduce the cost of amplitude amplification.
\begin{definition}[Rescaled Carleman problem]\label{defn:rescaled-carleman}
     Consider a nonlinear ODE system of the form $\dt{\uu} = F_1\uu + F_M \uu^{\otimes M}$ as in \cref{defn:nonlin-ode-problem}. Then, using a variable transformation in the form of a rescaling $\widetilde{\uu}= \uu/\gamma$ with $\gamma>0$, we obtain another system in the rescaled variable 
    \begin{equation}
    \label{eq:resc_ODE}
        \dt{\widetilde{\uu}} = \widetilde{F}_1 \widetilde{\uu} + \widetilde{F}_M \widetilde{\uu}^{\otimes M} , \quad 
    \end{equation}
    with $\widetilde{F}_1 = F_1$ and $\widetilde{F}_M = \gamma^{M-1} F_M$.
\end{definition}
This allows us to improve the measurement probability in the following sense.
\begin{lemma}[Measurement probability of the rescaled Carleman problem]\label{lem:rescaled-carleman}
    Using the rescaling in \cref{defn:rescaled-carleman}, using a scaling factor $\gamma\ge\norm{\uu_\mathrm{in}}$ and assuming dissipativity of the ODE, the probability to measure $\widetilde{\uu}=\widetilde{\yy}_1$ is given by 
    \begin{equation}\label{eq:rescaled-probability-1}
        P(\widetilde{\yy}_1(t))  =
        \frac{1 - \frac{\norm{\uu(t)}^2}{\gamma^2}}{1 - \left(\frac{\norm{\uu(t)}}{\gamma}\right)^{2 N} } \ge \frac 1N. 
    \end{equation}
    \begin{proof}
    Using the rescaling $\gamma>0$, we obtain a new normalisation
    \begin{equation}
        \widetilde{V}_N = \sum_{l = 1}^N \left(\frac{\norm{\uu(t)}}{\gamma}\right)^{2l} = \norm{\widetilde{\uu}(t)}^2 \frac{1 - \norm{\widetilde{\uu}(t)}^{2 N}}
        {1 - \norm{\widetilde{\uu}(t)}^2} \, ,
    \end{equation}
    with $\norm{\widetilde{\uu}(t)} = {\norm{\uu(t)}}/{\gamma}$.
    Given dissipativity of the ODE, we have $\norm{\uu(t)}\le \norm{\uu_\mathrm{in}}$, so $\norm{\widetilde{\uu}(t)}\le 1$.
    In turn that implies
    \begin{equation}
        \frac{1 - \norm{\widetilde{\uu}(t)}^{2 N}}
        {1 - \norm{\widetilde{\uu}(t)}^2} \le N.
    \end{equation}
    The measurement probability to obtain $\widetilde{\yy}_1(t)$ is then
    \begin{equation}\label{eq:rescaled-probability}
        P(\widetilde{\yy}_1(t)) = \sum_{j_1=1}^n\left|\braket{y_1^{(j_1)}}{\widetilde{\yy}(t)}\right|^2 
        = \frac{\norm{\widetilde{\uu}(t)}^{2}}{\widetilde{V}_N} = \frac
        {1 - \norm{\widetilde{\uu}(t)}^2}{1 - \norm{\widetilde{\uu}(t)}^{2 N}} \ge \frac 1N \, .
    \end{equation}
\end{proof}
\end{lemma}
Therefore, using the parameter $\gamma$, we can adjust the probability to obtain $\widetilde{\yy}_1$.
Here we have taken $\gamma\ge\norm{\uu_\mathrm{in}}$, though the first expression does not depend on this assumption.
The probability is equal to $1/N$ if $\gamma=\norm{\uu_\mathrm{in}}=\norm{\uu(t)}$, and otherwise for $\gamma>\norm{\uu(t)}$ the probability is even better.
Thus the rescaling avoids the exponential (in $N$) suppression of the probability of obtaining the component of interest of the ODE problem, which occurs for $\norm{\uu(t)}>1$ without rescaling.

When we apply the rescaling above into \cref{eq:carleman-assembly} we obtain a linearised system in the rescaled solution vector with $\widetilde{A}_j^{(1)} = A_j^{(1)}$ and $\widetilde{A}^{(M)}_{j+M-1} = \gamma^{M-1} A^{(M)}_{j+M-1}$, and as a result we can write the rescaled Carleman linearisation as
\begin{equation}
\dt{\widetilde{\yy}} = \widetilde{\mathcal{A}}_N\widetilde{\mathbf{y}}, 
\end{equation}
where 
\begin{equation}\label{eq:resccarlmatrix}
\widetilde{\mathcal{A}}_N=
\begin{bmatrix}
A_1^{(1)} & 0          & \cdots & 0        & \gamma^{M-1}A^{(M)}_{M}& 0             & \cdots       & 0          \\
0         & A_2^{(1)}  & \cdots & 0        & 0          & \gamma^{M-1}A^{(M)}_{M+1} & 0            & \vdots     \\
\vdots    & 0          & \ddots &          &            & 0             & \ddots       & 0          \\
          &            & \ddots & \ddots   &            &               & \ddots       & \gamma^{M-1}A_N^{(M)}  \\
          &            &        &          & \ddots     &               &              & 0          \\
          &            &        &          & \ddots     & \ddots        &              & \vdots     \\
\vdots    &            &        &          &            & 0             & A_{N-1}^{(1)}& 0          \\
0         & 0          & \cdots &          &            & \cdots        & 0            & A_N^{(1)}  \\
\end{bmatrix} .
\end{equation}
We discuss the cost of an implementation of the rescaled dynamics in \cref{sec:ODEs}.

\subsection{Error bounds on rescaled solution}

Next, we present error bounds on the global and component-wise errors due to Carleman linearisation in \cref{lem:error_Glob} and \cref{lem:error_Comp}, where we make use of the rescaling technique outlined in the previous section.
The first lemma provides a bound on the overall error in the Carleman vector. The error bounds we present here are based on the $2$-norm.

\begin{lemma}[Global rescaled Carleman error]
\label{lem:error_Glob}
Consider the ODE from \cref{eq:nonlin-ode-problem} with its Carleman linearisation in \cref{eq:Carl_line} truncated at order $N$. Let  $F_1$ be dissipative, so that for $\lambda_{0}<0$  with $|\lambda_{0}| > \uin^{M-1}\|F_M\|$ and therefore $\|\widetilde{\mathbf{u}}_{\mathrm{in}}\| \geq \|\widetilde{\uu}(t)\|$ for $t>0$.  Then, the error in the rescaled solution as defined in \cref{lem:rescaled-carleman} is given by $\eta_j=\widetilde{\uu}^{\otimes j}-\widetilde{\mathbf{y}}_j$ at order $j\in[N]$ due to Carleman truncation $N>M\geq 2$  and a scaling factor $\gamma=\|\uu_{\mathrm{in}}\|$; $\widetilde{\uu}$ denotes the exact solution to the underlying ODE whereas $\widetilde{\yy}$ is the approximation due to Carleman truncation. Then, this error for any $j\in[N]$ is upper bounded by the overall error vector,
\begin{equation}
\|\mathbf{\eta}_j(t)\| \leq \|\mathbf{\eta}(t)\| \leq (M-1)\|F_M\| \|\uu_{\mathrm{in}}\|^{M-1}\frac{1-e^{N(\lambda_0 + \gamma^{M-1}\|F_M\|)t}}{|\lambda_0+\gamma^{M-1}\|F_M\||}. 
\end{equation}
\end{lemma}
The detailed proof is presented in \cref{app:proof_total_error}. 
Related results were given in Ref.~\cite{liu2020efficient} and Ref.~\cite{krovi2023improved}.
Neither included a general power for the nonlinearity, and were restricted to $M=2$.
Furthermore, we provide an exponential reduction in the Carleman order dependence due to the rescaling, i.e., $\|\eta\|\propto \|\mathbf{u}_{\mathrm{in}}\|^{M}$ in opposed to $\|\mathbf{u}_{\mathrm{in}}\|^N$.
Although Ref.~\cite{krovi2023improved} mentioned rescaling, it appears not to have been used in the error analysis.
If the rescaled form was being used in that work, then it would imply that $\|\mathbf{u}_{\rm in}\|$ would be equal to 1, so $\log(1/\|\mathbf{u}_{\rm in}\|)=0$ which results in $N$ being infinite in Eq.~(7.23) of Ref.~\cite{krovi2023improved}. 

A problem with using this form is that it does not go down with the Carleman order.
We aim to show that the error may be made arbitrarily small with higher-order Carleman approximations. We can provide tighter bounds when we consider the individual components of the Carleman vector, as in the following lemma.

\begin{lemma}[Component-wise Carleman error]
\label{lem:error_Comp}
Under the same setting as in \cref{lem:error_Glob}, and $j \in [N]$, the Carleman error for each individual component of $\mathbf{\eta}_j$ satisfies 
\begin{equation}
 \|\eta_j(t)\| \leq \left( \frac{\|\uu_{\mathrm{in}}\|}{\gamma} \right)^j R^k  f_{j,k,M}\left(|\lambda_0|t\right),\qquad j\in \Omega_k
\end{equation}
where
\begin{equation}\label{eq:function-from-nested-integrals}
    f_{j,k,M}(\tau) = 1-\frac{(M-1) \, \Gamma(k+j/(M-1))}{(k-1)! \, \Gamma(j/(M-1))}\sum_{\ell=0}^{k-1} (-1)^\ell \binom{k-1}{\ell} \frac{e^{-(\ell M-\ell+j) \tau}}{\ell M-\ell+j} \, ,
\end{equation}
for  $k\in \{1, 2,\cdots, \lceil\frac{N}{M-1} \rceil\}$ and $k$ is determined so that for any $j$, we have $k$ whenever $j$ falls into the index set $j\in \Omega_k$ with
\begin{equation}
 \Omega_k \coloneqq \{N-k(M-1)+1,\ldots, N+(k-1)(1-M)\}.   
\end{equation}
In particular, for $k=\lceil N/(M-1) \rceil$ we have
\begin{equation}
\label{eq:error1}
\|\eta_1(t)\| \leq \frac{\|\uu_{\mathrm{in}}\|}{\gamma}  R^{ \lceil\frac{N}{M-1} \rceil} f_{1,\lceil N/(M-1) \rceil,M}\left(|\lambda_0|t\right). 
\end{equation}
\end{lemma}
The proof of \cref{lem:error_Comp} can be found in \cref{app:proof_comp}.
The function $f_{j,k,M}(\tau)$ is monotonically decreasing with $k$, and in particular $f_{j,k,M}(\tau)\le f_{j,1,M}(\tau) = 1-e^{-j\tau}$ (see \cref{app:proof_comp}).
This result does not depend on the choice of rescaling $\gamma$.
There is a factor of $1/\gamma^j$ in the definition of $\eta_j$, so the result is effectively independent of the choice of rescaling.
Moreover, $\|\eta_1(t)\|$ gives the error in the desired component at the end, and shows that the error in $\uu$ is proportional to $\|\uu_{\mathrm{in}}\|$.

A similar result was provided in Ref.~\cite{an2022efficient} without using the rescaling, though that does not affect the result for the error.
We give a significant improvement over the result in Ref.~\cite{an2022efficient} by evaluating the nested integrals to give the function $f_{j,k,M}(\tau)$, whereas the result in Ref.~\cite{an2022efficient} just corresponds to replacing $f_{j,k,M}(\tau)$ with its upper bound of 1.

We can use \cref{lem:error_Comp} to solve for a lower bound on $N$ for a given allowable error.
In practice, we are interested in the error in the solution relative to $\uin$ rather than $\gamma$, so we aim to bound $\|\eta_{1}(t)\|\gamma/\uin$.
Given a maximum allowable error $\varepsilon$, we then require
\begin{align}
 \varepsilon &\geq R^{ \lceil\frac{N}{M-1} \rceil} f_{1,\lceil N/(M-1) \rceil,M}\left(|\lambda_0|t\right) \geq R^{ \lceil\frac{N}{M-1} \rceil}\, .
\end{align}
It is therefore sufficient to choose $N$ as
\begin{equation}\label{eq:N-for-error}
\Big\lceil\frac{N}{M-1} \Big\rceil \geq \frac{\log\left(1/\varepsilon\right)}{\log\left(1/R\right)},
\end{equation}
or
\begin{equation}\label{eq:N-for-error2}
    N= (M-1) \left \lceil     \frac{\log\left(1/\varepsilon\right)}{\log\left(1/R\right)}
    \right\rceil - (M-2) \, .
\end{equation}
We can also numerically solve for $N$, by using the exact expression for $f_{j,k,M}(\tau)$ given in \cref{eq:function-from-nested-integrals}.
That will give a tighter lower bound on $N$, but there is not a closed-form expression.

\section{Solution of the linearised system of ordinary differential equations using a truncated Taylor series}\label{sec:ODEs}
Next, we describe how to solve the system of ODEs that results from the Carleman mapping applied onto the nonlinear system. 
The most simple way to solve the system of ODEs is to apply the first-order method for time discretisation known as the explicit Euler method. Upon application of the Euler method, there is a linear system of equations that can be solved. Here, this is a quantum linear system problem~(QLSP), as the solution is encoded in a quantum state. 
In what follows, we aim to solve the linear system by a more sophisticated method than explicit Euler.
The main drawback of the forward Euler method is low accuracy since it is a first-order method, meaning finer time discretisation is required to achieve a required precision. 
As a result, the dependence of the complexity for solving the QLSP is quadratic in the solution time, and there is a near-linear factor in the inverse error \cite{liu2020efficient,an2022efficient}.

Here, we follow the procedure outlined in Ref.~\cite{berry2022quantum}, which allows us to obtain an algorithm that has complexity near-linear in time and logarithmic in the inverse error.
The solution of a time-independent ODE system
\begin{equation}
\dt{\mathbf{u}(t)} = A\mathbf{u}(t) , 
\end{equation}
may be approximated by $\uu_K(t) =  W_K(t,t_0) \uu(t_0)$, with
\begin{align}
    W_K(t,t_0) &:=\sum_{\ell=0}^K \frac{(A\Delta t)^\ell}{\ell !} .
\end{align}
This is a Taylor series truncated at order $K$.
The error in the solution due to time propagation can be bounded as
\begin{equation}
\| \mathbf{u}_K(t) - \mathbf{u}(t) \| \in  \mathcal{O}\left(
\frac{(\norm{A}\Delta t)^{K+1}}{(K+1)!}\norm{\uu(t_0)} 
\right).
\end{equation}
We aim to solve \cref{eq:CarlODE} where the vector $\mathbf{u}(t)$ is mapped to a rescaled vector $\widetilde\yy(t)$ and $A$ is the rescaled Carleman matrix $\widetilde{\mathcal{A}}_{N}$ truncated at order $N$.

Following Theorem~2 in \cite{berry2022quantum}, there exists a quantum algorithm that can provide an approximation $\ket{\hat{\mathbf{y}}}$ of the solution $\ket{\widetilde{\mathbf{y}}(T)}$ satisfying $\left\| \ket{\hat{\mathbf{y}}} - \ket{\widetilde{\mathbf{y}}(T)}\right\| \le \varepsilon y_{\max}$. To do so, we require that $\carlmatrix$ has non-positive logarithmic norm and we have the oracles 
$U_y$ to prepare the
initial state and block encoding of $\carlmatrix$ via
$U_{\widetilde{\mathcal{A}}_N}$ with $\bra{0}U_{\widetilde{\mathcal{A}}_N}\ket{0} =
\widetilde{\mathcal{A}}_N/\lambda_{\widetilde{\mathcal{A}}_N}$.
Then, to achieve the desired accuracy, the average number of calls to $U_y$ and $U_{\widetilde{\mathcal{A}}_N}$ needed are
\begin{align}
\label{eq:Avg_calls}
    U_y: & \quad \order{ \widetilde{\mathcal{R}} \lambda_{\widetilde{\mathcal{A}}_N} T \log(\frac{1}{\varepsilon})   } \\
    \label{eq:Avg_calls2}
    U_{\widetilde{\mathcal{A}}_N}: & \quad \order{ \widetilde{\mathcal{R}}\lambda_{\widetilde{\mathcal{A}}_N} T \log(\frac{1}{\varepsilon})  \log(\frac{\lambda_{\widetilde{\mathcal{A}}_N} T}{\varepsilon})} \, .  
\end{align}
Furthermore, the number of additional elementary gates scales as
 \begin{equation}
 \label{eq:ode-extra-gates}
     \order{ \widetilde{\mathcal{R}}\lambda_{\widetilde{\mathcal{A}}_N} T \log(\frac{1}{\varepsilon})  \log^2\left(\frac{\lambda_{\widetilde{\mathcal{A}}_N} T}{\varepsilon}\right)} \, .
\end{equation}
In these expressions
\begin{align}\label{eq:complexity-R-factor}
\widetilde{\mathcal{R}} &\ge \frac{y_{\max}}{\left\| \widetilde{\mathbf{y}}(T) \right\|}\\
 y_{\max} &\ge \max_{t\in[0,T]}\left\| \widetilde{\mathbf{y}}(t) \right\|.
\end{align}

The stability requirement on the ODE to use the solver as in Ref.~\cite{berry2022quantum} is that the logarithmic norm of the matrix is non-positive (similar to Ref.~\cite{krovi2023improved}).
That norm is given by the eigenvalues of $(\widetilde{\mathcal{A}}_N+\widetilde{\mathcal{A}}_N^\dagger)/2$.
The eigenvalues of that matrix can be bounded via the block form of the Gershgorin circle theorem.
That is equivalent to the usual Gershgorin circle theorem, except using the spectral norms of the off-diagonal blocks.
For example, see Theorem 2 of Ref.~\cite{Feingold1962BlockDD}, or Ref.~\cite{VANDERSLUIS1979265}.

For $(\widetilde{\mathcal{A}}_N+\widetilde{\mathcal{A}}_N^\dagger)/2$ we obtain rows with $A_j^{(1)}$ and $\gamma^{M-1}A_j^{(M)}/2$ (for $j\ge M$) and $\gamma^{M-1}A_{j+M-1}^{(M)}/2$ (for $j+M-1\le N$).
Now $\|A_{j+M-1}^{(M)}\|\le j\|F_M\|$, so the sum of the norms of the off-diagonal blocks is at most, for $j\ge M$ and $j+M-1\le N$,
\begin{equation}
    \|A_{j+M-1}^{(M)}\| + \|A_{j}^{(M)}\| \le j\|F_M\| + (j-M+1)\|F_M\| = (2j-M+1)\|F_M\| \, .
\end{equation}
In the case $j< M$ but $j+M-1\le N$ then we get $j\|F_M\|$.
If $j+M-1 > N$ but $j\ge M$ then we get $(j-M+1)\|F_M\|$.
Now the maximum eigenvalue of $[A_j^{(1)}+(A_j^{(1)})^\dagger]/2$ is $j\lambda_0$.
In that case the eigenvalues of $(\widetilde{\mathcal{A}}_N+\widetilde{\mathcal{A}}_N^\dagger)/2$ can be at most
\begin{equation}\label{eq:matcirc}
\begin{cases}
    j\lambda_0 + j\gamma^{M-1}\|F_M\|/2 \, , & 0<j< M \\
    j\lambda_0 + (2j-M+1)\gamma^{M-1}\|F_M\|/2\, , & j\ge M ~{\rm and} ~j\le N-M+1 \\
    j\lambda_0 + (j-M+1)\gamma^{M-1}\|F_M\|/2 \, , & N \ge j > N-M+1
\end{cases}
\end{equation}
We can then see that the eigenvalues will be non-positive given all three inequalities
\begin{align}
 \gamma^{M-1} &\le \frac{2|\lambda_0|}{\|F_M\|}\, , \\
    \gamma^{M-1} &\le \frac{|\lambda_0|}{[1-(M-1)/(2(N-M+1))]\|F_M\|}\, , \\
    \gamma^{M-1} &\le \frac{|\lambda_0|}{[1-(M-1)/N]\|F_M\|} \, .
\end{align}
Provided $N\ge 2(M-1)$ (as would normally be the case) the middle inequality would imply the other two.
In all cases we can satisfy these inequalities using
\begin{equation}
    \gamma^{M-1} \le \frac{|\lambda_0|}{\|F_M\|}  = \frac{\|\mathbf{u}_{\mathrm{in}}\|^{M-1}}{R} \, ,
\end{equation}
or
\begin{equation}
    \gamma \le \frac{\|\mathbf{u}_{\mathrm{in}}\|}{R^{1/(M-1)}}, \, 
\end{equation}
where we used the definition of $R$ from \cref{eq:R} in the equality above. Reference~\cite{berry2022quantum} argues that for cases where the solution does not decay significantly, $\mathcal{R}\in\order{1}$.
Here, we consider dissipative dynamics without driving, so $\mathcal{R}$ may be large.
That is less of a problem for driven equations.
We expect that our methods can be applied to driven equations as well, but the error analysis is considerably more complicated so we leave it as a problem for future work. 

We can construct the block encoding of the Carleman matrix $\widetilde{\mathcal{A}}_N$ in terms of the block encoding of $F_1$ and $F_M$, as discussed in \cref{app:block_enc}.
Denoting the values of $\lambda$ for $F_1$ and $F_M$ by $\lambda_{F_1}$ and $\lambda_{F_M}$ respectively, the value of $\lambda$ for $\widetilde{\mathcal{A}}_N$ is
\begin{align}
\label{eq:lambda-rescaled-specific}
\lambda_{\widetilde{\mathcal{A}}_N} \le N \lambda_{F_1} + (N-M+1)\gamma^{M-1}\lambda_{F_M}\, .
\end{align}
This expression easily follows from expressing $\widetilde{\mathcal{A}}_N$ as a sum, and the value of $\lambda$ being the sum of the values of $\lambda$ in the sum.
Since $\widetilde{\mathcal{A}}_N$ includes $A_j^{(1)}$ up to $A_N^{(1)}$, and $A_N^{(1)}$ is a sum of $N$ operators with identity tensored with $F_1$, we obtain the term $N\lambda_{F_1}$ above.
Similarly, we have $\gamma^{M-1}A_j^{(M)}$ up to $\gamma^{M-1}A_N^{(M)}$, and $A_N^{(M)}$ is a sum of $N-M+1$ operators with with $F_M$, giving the $(N-M+1)\gamma^{M-1}\lambda_{F_M}$ term.

If we choose $\gamma^{M-1} = {|\lambda_0|}/{\|F_M\|}$ as above, then
\begin{align}
N \norm{F_1} > (N-M+1)\gamma^{M-1}\norm{F_M} \, .
\end{align}
In typical cases we would expect that $\lambda_{F_1}\propto \norm{F_1}$ and $\lambda_{F_M}\propto \norm{F_M}$.
That would imply
\begin{align}\label{eq:lambdabound}
    \lambda_{\widetilde{\mathcal{A}}_N} \lesssim 2N \lambda_{F_1} \, .
\end{align}
Note that the scaling has not increased the value of $\lambda$ by more than a constant factor.
Note that this is assuming that the $\lambda$-values and norms in the block encoding are comparable, so it is possible it could be violated if the block encoding of $F_M$ is inefficient, so $\lambda_{F_M}$ is much larger than $\norm{F_M}$.

Now for $\mathcal{R}$ we have $y_{\max}$ which considers the maximum norm that the vector can assume along the entire time evolution. Since we are working with a dissipative problem the maximum occurs at $t=0$.
First we consider the case without the scaling for comparison.
To compute the norm $\norm{\yy(0)}$, note that it is the vector resulting from the Carleman mapping, i.e., $\yy(0) = [\uu_{\mathrm{in}}, \uu_{\mathrm{in}}^{\otimes2}, \ldots, \uu_{\mathrm{in}}^{\otimes N}]^T$, so
\begin{align}
    \norm{\yy(0)}^2 
    = \|\mathbf{u}_{\mathrm{in}}\|^2 \frac{1-\|\mathbf{u}_{\mathrm{in}}\|^{2N}}{1-\|\mathbf{u}_{\mathrm{in}}\|^2} ,
\end{align}
as in Eq.~\eqref{eq:normalis}.
Similarly for the value of the norm at time $T$,
\begin{equation}
\norm{\yy(T)}^2 =  \|\uu(T)\|^2{\frac{1-\|\uu(T)\|^{2N}}{1-\|\uu(T)\|^{2}}} . 
\end{equation}
Therefore
\begin{align}
\mathcal{R}&\geq\frac{y_{\max}}{\norm{\yy(T)}}\nn 
&= \left[\frac{\left(1-\|\mathbf{u}_{\mathrm{in}}\|^{2N}\right)\left(1-\|\uu(T)\|^{2}\right)}{\left(1-\|\mathbf{u}_{\mathrm{in}}\|^2\right)\left(1-\|\mathbf{u}(T)\|^{2N}\right)}\right]^{1/2}\frac{\uin}{\norm{\uu(T)}}.
\end{align}

Moreover, the above complexity is in order to obtain the full Carleman vector.
The quantity $\mathcal{R}$ corresponds to an inverse amplitude for obtaining the state at the final time, so tells us how many steps of amplitude amplification are needed in the algorithm.
In practice, we want only $\uu(T)$ rather than the full vector.
That implies a further factor in the complexity of $\norm{\yy(T)}/\norm{\uu(T)}$, corresponding to the inverse amplitude for obtaining the component of the Carleman vector containing the solution.
That gives a factor in the complexity of
\begin{align}\label{eq:Rcost}
\frac{\norm{\yy(T)}}{\norm{\uu(T)}}\mathcal{R}&\geq
 \left[\frac{\left(1-\|\mathbf{u}_{\mathrm{in}}\|^{2N}\right)}{\left(1-\|\mathbf{u}_{\mathrm{in}}\|^2\right)}\right]^{1/2}\frac{\uin}{\norm{\uu(T)}}.
\end{align}
From the equation above we can see how $\mathcal{R}$ grows exponentially in $N$ for $\uin>1$.

Now with the rescaling, we simply divide each $\mathbf{u}_{\mathrm{in}}$ or $\uu(T)$ by $\gamma$.
That gives us
\begin{align}\label{eq:Rcost2}
\frac{\norm{\widetilde\yy(T)}}{\norm{\widetilde\uu(T)}}\widetilde{\mathcal{R}}&\geq
 \left[\frac{\left(1-\|\mathbf{u}_{\mathrm{in}}\|^{2N}/\gamma^{2N}\right)}{\left(1-\|\mathbf{u}_{\mathrm{in}}\|^2/\gamma^2\right)}\right]^{1/2}\frac{\uin}{\norm{\uu(T)}}.
\end{align}
With the choice $\gamma^{M-1} = {|\lambda_0|}/{\|F_M\|}$, we obtain
\begin{align}\label{eq:Rcost3}
\frac{\norm{\widetilde\yy(T)}}{\norm{\widetilde\uu(T)}}\widetilde{\mathcal{R}}&
\geq
 \frac{1}{\sqrt{1-R^{2/(M-1)}}}
 \frac{\uin}{\norm{\uu(T)}}.
\end{align}
We then can see that the amplitude amplification cost can be exponentially reduced when $\uin > 1$.
We could also use $\gamma=\uin$ to give
\begin{align}
\label{eq:resc_R}
\frac{\norm{\widetilde\yy(T)}}{\norm{\widetilde\uu(T)}}\widetilde{\mathcal{R}} \geq \sqrt{N}\frac{\uin}{\norm{\uu(T)}} \, ,
\end{align}
but that bound is looser for realistic parameters.
 
A further consideration is the relation between the relative error in the solution for $\widetilde{\yy}(T)$ and that for $\uu(T)$.
The complexity of the solution for the ODE solver is in terms of the former, whereas we need to bound the relative error in $\uu(T)$.
We have the error upper bounded by (with hats used to indicate results given by the linear equation solver)
\begin{align}\label{eq:relerr1}
    \norm{\hat{\uu}(T)-\uu(T)} &\le \gamma \norm{\hat{\yy}(T)-\yy(T)} \nn
    &\le \gamma \varepsilon y_{\max} \nn
    &\le \gamma \left[\frac{\left(1-\|\mathbf{u}_{\mathrm{in}}\|^{2N}/\gamma^{2N}\right)}{\left(1-\|\mathbf{u}_{\mathrm{in}}\|^2/\gamma^2\right)}\right]^{1/2}\frac{\uin}{\gamma} \nn
    &\le \varepsilon \uin \frac{1}{\sqrt{1-R^{2/(M-1)}}} \, .
\end{align}
In the second line we have assumed that the ODE solver has given the solution for $\yy(T)$ to within error $\varepsilon y_{\max}$.
This shows that the relative error in $\uu(T)$ is the same as that for $\widetilde{\yy}(T)$, up to a factor of $1/\sqrt{1-R^{2/(M-1)}}$ which should be close to 1.
We can also use the simpler but looser upper bound
\begin{align}\label{eq:relerr2}
    \norm{\hat{\uu}(T)-\uu(T)} &\le \varepsilon \uin \sqrt{N} \, ,
\end{align}
which is obtained by noting that the expression in the square brackets in the third line of Eq.~\eqref{eq:relerr1} is upper bounded by $N$.

We can now use the ODE solver given in Ref.~\cite{berry2022quantum} in combination with our rescaling technique to provide our quantum algorithm for \cref{defn:nonlin-ode-problem}.
\begin{lemma}[Complexity of solving ODE]\label{lem:ode-complexity} There is a algorithm to solve the nonlinear ODE from \cref{eq:nonlin-ode-problem} i.e., to produce a quantum state $\ket{\hat{\uu}(T)}$ encoding the solution such that
$\norm{\hat{\uu}(T)-\uu(T)}\le \varepsilon \uin$,
using an average number
\begin{equation}
    \order{
    \frac{1}{\sqrt{1-R^{2/(M-1)}}}\frac{\uin}{\norm{\uu(T)}} \lambda_{F_1} T N\log({\frac{N}{\varepsilon}}) \log(\frac{ N\lambda_{F_1}T}{\varepsilon})},
\end{equation}
of calls to oracles for $F_1$ and $F_M$, 
\begin{equation}
\label{eq:nonODE_comp}
    \order{
    \frac{1}{\sqrt{1-R^{2/(M-1)}}}\frac{\uin}{\norm{\uu(T)}} \lambda_{F_1} T N^{2}\log({\frac{N}{\varepsilon}}) },
\end{equation}
calls to oracles for preparation of $\mathbf{u}_{\mathrm{in}}$, and 
\begin{equation}
    \order{
    \frac{1}{\sqrt{1-R^{2/(M-1)}}}\frac{\uin}{\norm{\uu(T)}} \lambda_{F_1} T N^{2}M \log({\frac{N}{\varepsilon}}) \log(\frac{ N\lambda_{F_1}T}{\varepsilon})^2\log n},
\end{equation}
additional gates for dimension $n$, with
\begin{equation}
    N = \order{(M-1) \frac{\log\left(1/\varepsilon\right)}{\log\left(1/R\right)}}.
\end{equation}
We require that $R<1$ and assume that $\lambda_{F_M}/\|F_M\| = \mathcal{O}(\lambda_{F_1}/\|F_1\|)$ for the block encodings of $F_1$ and $F_M$.
\end{lemma}
\begin{proof}
The main step to derive our quantum algorithm is first to apply the Carleman linearisation in the rescaled nonlinear ODE problem, which is given in \cref{eq:resc_ODE}. We then have a linear ODE problem with the Carleman matrix of order $N$, denoted $\widetilde{\mathcal{A}}_N$. We can then apply the ODE solver given in Ref.~\cite{berry2022quantum} to this equation.

There are then a number of considerations needed to give the overall complexity.
\begin{itemize}
    \item We need to multiply by a further factor of ${\norm{\widetilde\yy(T)}}/{\norm{\widetilde\uu(T)}}$ to obtain the correct component of the solution containing the approximation of $\uu(T)$.
    The product of that with $\widetilde{\mathcal{R}}$ is given above in Eq.~\eqref{eq:Rcost3}.
    \item The value of $\lambda_{\widetilde{\mathcal{A}}_N}$ is given above in Eq.~\eqref{eq:lambdabound} under the assumption $\lambda_{F_M}/\|F_M\| = \mathcal{O}(\lambda_{F_1}/\|F_1\|)$, which gives $\lambda_{\widetilde{\mathcal{A}}_N}=\mathcal{O}(N \lambda_{F_1})$.
    \item The matrix $\widetilde{\mathcal{A}}_N$ can be block encoded with  $\order{1}$ calls to the oracles for $F_1$, $F_M$.
    There is an extra $\order{N}$ factor for the number of calls to $\mathbf{u}_{\mathrm{in}}$.
    The implementation of the oracles is explained in Appendix \ref{app:block_enc}.
    \item The choice of the Carleman order $N$ in order to obtain a sufficiently accurate solution is given in Eq.~\eqref{eq:N-for-error2}.
    The error from the Carleman truncation can be chosen to be a fraction of the total allowable relative error $\varepsilon$ here, which is accounted for using the order notation for $N$.
    \item The solution for the ODE can be given to relative error $\varepsilon/\sqrt{N}$.
    According to Eq.~\eqref{eq:relerr2} that will ensure that the relative error in $\uu(T)$ obtained is $\varepsilon$ as required.
    It is for this reason that we have replaced the $1/\varepsilon$ in the complexity for the ODE solver with $N/\varepsilon$.
\end{itemize}
For the additional elementary gates, the block encoding as in Appendix \ref{app:block_enc} requires a factor of $\order{NM\log n}$ for swapping target registers into the appropriate location.
That is a factor on the number of block encodings of $\widetilde{\mathcal{A}}_N$.
Moreover, Ref.~\cite{berry2022quantum} gives a log factor to account for the complexity of correctly giving the weighting in the Taylor series.
For simplicity we give the product of these factors, but these factors are for different contributions to the complexity and we could instead give a more complicated expression with the maximum of $NM\log n$ and the logarithm.
\end{proof}

We can compare our quantum algorithm performance with what is given in Theorem 8 of Ref.~\cite{krovi2023improved} for the case $M=2$.
The complexity given in that theorem can be simplified to the situation we consider by removing the driving term and replacing $\norm{A}$ with $\lambda_{\mathcal{A}_N}$.
Then the complexity from \cite{krovi2023improved} is
\begin{equation}
\label{eq:krovi}
\order{
    \frac{\uin}{\norm{\uu(T)}}  \lambda_{F_1} TN \text{~poly}\left(N,\log{\left(\frac{1}{\varepsilon}\right)},\log{(TN\lambda_{F_1})}\right)} \, .
\end{equation}
The speedup is unclear because the complexity in that work is given in terms of poly factors.
That work appears to be assuming a rescaling in order to avoid complexity exponential in $N$, but by assuming $\uin=1$.
The problem is that they give a formula for $N$ as
\begin{equation}
\label{eq:N_krovi}
N = \left \lceil \frac{2\log{\left(T\|F_2\|/\delta \|\uu(T)\|\right)}}{\log\left(1/\uin\right)}
    \right\rceil \, .
\end{equation}
Using $\uin=1$ in that formula gives infinite $N$.
In contrast, here we have given the rescaling explicitly and given a working formula for $N$.

\section{Application to the quantum nonlinear PDE problem}\label{sec:nonlin-pde-problem}

\subsection{Complexity of the quantum algorithm}
We now demonstrate our techniques applied to the nonlinear PDE \cite{an2022efficient}
\begin{equation}
\label{eq:pde-example}
\partial_t u (\xx,t) = D \Delta  u(\xx,t) + c  u(\xx,t)+ b  u^{M} (\xx,t), 
\end{equation}
for some diffusion coefficient $D\ge0$ and constants $c,b\in\mathbb{R}$.
As a simple means of discretisation we consider finite differences with periodic boundary conditions, which leads to a vector-valued ODE that approximates the dynamics in \cref{eq:pde-example}. We go beyond the two-point stencil demonstrated in Ref.~\cite{an2022efficient} and apply higher-order finite differences similar to Ref.~\cite{childs2021high} for the linear case.

We discretise a $d$-dimensional space in each direction with uniformly equidistant grid points. As a result, we obtain a nonlinear system of ODEs as in \cref{eq:nonlin-ode-problem} with $n$ grid points in total, or $n^{1/d}$ in each direction. Moreover, we consider the width of the simulation region to be 1 in each direction, so $x_j\in [0,1]$, for simplicity.

The linear operator $F_{1}$ resulting from the spatial discretisation of our PDE is given by
\begin{equation}
\label{eq:F1}
F_{1}=DL_{k,d}+c\identity^{\otimes d},    
\end{equation}
where $\identity$ is the $n^{1/d}\times n^{1/d}$ identity matrix, and 
\begin{equation}
\label{eq:d_Laplac}
L_{k,d} =  \sum_{\mu=1}^d \identity^{\otimes(\mu-1)} \otimes L_k \otimes \identity^{\otimes(d-\mu)}.
\end{equation}
The operator $L_{k,d}$ above for the discretised Laplacian in dimension $d$ is constructed from the sum of the discretised Laplacians in one dimension, $L_k$.
Here, $k$ is the order, so the truncation error scales as the inverse grid spacing to the power of $2k-1$, and $2k+1$ stencil points are used.

A Laplacian in one dimension with a $k$th order approximation and periodic boundary conditions can be expressed in terms of weights $a_j$ as (see Ref.~\cite{childs2021high})
\begin{equation}
\label{eq:ho-laplacian-def}
    L_k = n^{2/d}\left(a_0I +  \sum_{j=1}^k a_j (S^{j} + S^{-j})\right), 
\end{equation}
where $S$ is a $n^{1/d} \times n^{1/d}$ matrix, where the entries are $S_{i,j} = \delta_{i,j+1\mod{n^{1/d}}}$; $S$ is also known as a circulant matrix.
Note that for a total of $n^{1/d}$ grid points and a region size of $1$ in each direction, the grid spacing is $1/n^{1/d}$. The method to obtain the coefficients $a_j$ for the Laplacian operator given in \cref{app:fd-coefficients} guarantees that 
\begin{equation}
\label{eq:period}
 a_0 + 2\sum_{j=1}^ka_j=0 \, ,   
\end{equation}
and we provide the coefficients for $1\le k\le 5$ in~\cref{tab:tableL} (these are from \cite{costa2019quantum}). Moreover, this procedure leads to a truncation error in the representation of the Laplacian operator which scales as (for the 2-norm) \cite{kivlichan2017bounding,childs2021high}
\begin{equation}\label{eq:trunc-error-laplacian}
    \order{C(u,k)\sqrt{n} \left(\frac{e}{2}\right)^{2k}n^{(-2k+1)/d}} \, ,
\end{equation}
where $C(u,k)$ is a constant depending on the $(2k+1)$st spatial derivative in each direction
\begin{equation}
\label{eq:Coef_C}
C(u,k) = \sum_{j=1}^d\left|\frac{d^{2k+1}u}{dx_j^{2k+1}}\right| \, .
\end{equation}
This expression is obtained from that in Refs.~\cite{kivlichan2017bounding,childs2021high} by adding the errors for derivatives in each direction.

\begin{table}[ht]
\begin{center}
\begin{tabular}{c|c}
Order $k$ & ~ Coefficients $a_{0}$ to $a_k$\\
\hline
1 & -2, 1 \\
2 & -5/2, 4/3, -1/12  \\
3 & -49/18, 3/2, -3/20, 1/90\\
4 & -205/72, 8/5, -1/5, 8/315, -1/560\\
5 &  -5269/1800, 5/3, -5/21, 5/126, -5/1008, 1/3150
\end{tabular}
\end{center}
\caption{\label{tab:tableL} Central finite difference coefficients for approximating a second derivative in one dimension \cite{costa2019quantum}.}
\end{table}

We also have the matrix $F_{M}$ resulting from the spatial discretisation of the nonlinear part $b u^{M} (\xx,t)$ that is a rectangular matrix $F_{M}$,
\begin{equation}
F_{M}: \mathbb{R}^{n^{M}} \rightarrow  \mathbb{R}^{n},  
\end{equation}
operating on the vector $\mathbf{u}^{\otimes M}$, as given in \cref{eq:tensM}. Since we are only interested in the components $u_i^{M}$ from $\mathbf{u}^{\otimes M}$, where $i=1,2,\cdots n$, $F_M$ is a one sparse matrix with the non-zero components given by $b$.
Hence $\|F_{M}\|=\|F_{M}\|_{\max}=|b|$. In the case $M=2$, where
\begin{equation}
\mathbf{u}^{\otimes2}= (u_1^2,u_1u_2,\dots,u_1u_{n},u_2u_1,u_2^2,\dots,u_{n}u_{n -1},u_{n^2})^T \in \mathbb{R}^{n^{2}},
\end{equation} 
we can express $F_2$ as 
\begin{equation}\label{eq:FM-definition}
    [F_{2}]_{pq} = \begin{cases}
        b, &  q = p + (p-1)n \\ 
        0, & \mathrm{otherwise}.
    \end{cases}
\end{equation}

Returning to the Laplacian operator with periodic boundary conditions, we see that \cref{eq:ho-laplacian-def} is a circulant matrix, so its eigenvalues are given by
\begin{align}
\label{eq:eig_Lk}
 \lambda_\ell(L_k) &= n^{2/d}\left[a_0 +\sum_{j=1}^{k}a_j\left(\omega^{\ell j} + \omega^{-\ell j} \right) \right] \nn
 &=n^{2/d}\left[a_0 +2\sum_{j=1}^{k}a_j\cos\left(\frac{2\pi \ell j}{n^{1/d}}\right)\right], \quad \ell\in [n^{1/d}-1]
\end{align}
where $\omega=e^{i 2\pi/n^{1/d}}$. Since the $a_j$, with $j=0,1,\cdots,k,$ satisfy the condition in \cref{eq:period}, we see that for $\ell=0$, $\lambda_0=0$ gives the maximum eigenvalue  and $\lambda_\ell<0$ for $\ell\neq0$. Moreover, using the triangle inequality in \cref{eq:ho-laplacian-def} we see that
\begin{equation}
 \|L_k\| \leq n^{2/d}\left(|a_0| + 2\sum_{j=1}^k|a_j|\right).   
\end{equation}
By a simple application of Gershgorin's circle theorem, one may obtain the asymptotic bound (Lemma 2 in Ref.~\cite{childs2021high}, Lemma 6 in Ref.~\cite{kivlichan2017bounding}), 
\begin{equation}
    \norm{L_k} \le n^{2/d}\frac{4\pi^2}{3}.
\end{equation}

From the eigenvalues of $L_k$ we can then determine the eigenvalues of $F_1$ (which is symmetric so equal to $(F_1+F_1^\dagger)/2$) as defined in \cref{eq:F1} as
\begin{equation}
\label{eq:evalues_F1}
\lambda_{\vec\ell}(F_1) = c + Dn^{2/d}\sum_{p=1}^d\left[a_0 +2\sum_{j=1}^{k}a_j\cos\left(\frac{2\pi \ell_p j}{n^{1/d}}\right)\right].
\end{equation}
As discussed above the maximum eigenvalue of $L_k$ is 0 with periodic boundary conditions (it can be negative for non-periodic boundary conditions).
For Carleman linearisation to be successful, we require $R<1$ and, in particular, $\lambda_0 <0$; see \cref{sec:carleman}.
Since the maximum eigenvalue of $F_1$ is $c$, we choose negative $c$ such that the overall dynamics becomes dissipative and satisfies the stability condition $R<1$.
Moreover, we obtain the following bounds
\begin{equation}
\label{eq:norm_F1}
 \|F_{1}\| \le \abs{c} +  dD n^{2/d} \left( |a_0| + \sum_{j=1}^k |a_j| \right) \le \abs{c} + dDn^{2/d} \frac{4\pi^2}{3}.
\end{equation}

The value of $\lambda_{F_1}$ for the block encoding of $F_1$ can be determined in a similar way.
The block encoding can be implemented by a linear combination of unitaries of the identity and powers of the circulant matrices $S$.
The value of $\lambda_{F_1}$ is then exactly equal to the sum
\begin{equation}
\label{eq:norm_F1b}
\lambda_{F_1} = \abs{c} +  dD n^{2/d} \left( |a_0| + \sum_{j=1}^k |a_j| \right) \le \abs{c} + dDn^{2/d} \frac{4\pi^2}{3}.
\end{equation}
Similarly, since $F_M$ is one-sparse it can be easily block encoded with a value of $\lambda_{F_M}$ equal to its norm of $|b|$.
We then obtain the value of $\lambda$ for the complete block encoding as
\begin{align}
\label{eq:lam_pde}
 \lambda_{\widetilde{\mathcal{A}}_N} &= N \lambda_{F_1} + (N-M+1)\gamma^{M-1}\lambda_{F_M}\nn 
 &\le N \left(\abs{c} + dDn^{2/d} \frac{4\pi^2}{3}\right) + (N-M+1)\gamma^{M-1}|b| \, .  
\end{align}
The condition on the dissipativity of the ODE $R<1$ implies that
\begin{equation}
N \left(\abs{c} + dDn^{2/d} \frac{4\pi^2}{3}\right) > (N-M+1)\gamma^{M-1}|b| \, . 
\end{equation}

Given these results for the discretisation of the PDE, we can use
\cref{lem:ode-complexity} to provide the following corollary.

\begin{corollary}[Complexity of solving a dissipative reaction-diffusion PDE]\label{cor:pde-complexity}
 There is a quantum algorithm to solve the nonlinear PDE in \cref{eq:pde-example} i.e., to produce a quantum state $\ket{\hat{\uu}(T)}$ encoding the solution such that
$\norm{\hat{\uu}(T)-\uu(T)}\le \varepsilon \uin$, using an average number
\begin{equation}
    \order{
    \frac{1}{\sqrt{1-R^{2/(M-1)}}}\frac{\uin}{\norm{\uu(T)}} (dDn^{2/d} + \abs{c}) T N\log({\frac{N}{\varepsilon}}) \log(\frac{ N(dDn^{2/d} + \abs{c})T}{\varepsilon})},
\end{equation}
of calls to oracles for $F_1$ and $F_M$, as defined in \cref{eq:F1} and \cref{eq:FM-definition} respectively, 
\begin{equation}
\order{
    \frac{1}{\sqrt{1-R^{2/(M-1)}}}\frac{\uin}{\norm{\uu(T)}} (dDn^{2/d} + \abs{c}) T N^{2}\log({\frac{N}{\varepsilon}}) },
\end{equation}
calls to oracles for preparation of $\mathbf{u}_{\mathrm{in}}$, and 
\begin{equation}
    \order{
    \frac{1}{\sqrt{1-R^{2/(M-1)}}}\frac{\uin}{\norm{\uu(T)}} (dDn^{2/d} + \abs{c}) T N^{2}M \log({\frac{N}{\varepsilon}}) \log(\frac{ N(dDn^{2/d} + \abs{c})T}{\varepsilon})^2\log n},
\end{equation}
additional gates, with
\begin{equation}
    N = \order{(M-1) \frac{\log\left(1/\varepsilon\right)}{\log\left(1/R\right)}}.
\end{equation}
We require that $R<1$, where $R$ is computed from the discretised input vector $\mathbf{u}_{\mathrm{in}}$.
\end{corollary}
\begin{proof}
We first discretise the reaction-diffusion problem in \cref{eq:pde-example} to the nonlinear ODE system with $n$ discretisation points.
We consider just the error in the solution of this ODE here, with the choice of $n$ to accurately approximate the solution of the PDE described below.
For this ODE we have an explicit bound for $\lambda_{F_1}$ given in Eq.~\eqref{eq:norm_F1b}, and can use it in the expressions in \cref{lem:ode-complexity}.
For this simple $F_M$ the values of $\lambda_{F_M}$ and $\|F_M\|$ are equal,
so the condition $\lambda_{F_M}/\|F_M\| = \mathcal{O}(\lambda_{F_1}/\|F_1\|)$ is satisfied.
\end{proof}

Note that for this result the oracles for $F_1$ and $F_M$ can be easily implemented in terms of calls to elementary gates, with logarithmic complexity in $n$ and linear complexity in $M$.
Powers of the circulant matrices can be implemented with modular addition, and $F_M$ can be implemented via equality tests between the copies it acts upon.
Note also that, apart from the $R<1$ condition, this complexity scales as $n^{2/d}$ up to logarithmic factors.
For $d\ge 3$ this complexity is sublinear in $n$.
This factor comes from the size of the discretised Laplacian, and is similar to that for quantum algorithms for linear PDEs.
If we had the factor of $\uin^{2N}$ as in Ref.~\cite{an2022efficient}, then because $\uin^2\propto n$ with the discretisation there would be a further factor of $n^N$ for the scaling with $n$, making the complexity far worse than that for a simple classical solver.

\subsection{Stability and discretisation}
\label{sec:stab}
Here we discuss conditions on nonlinear differential equations of the type in \cref{eq:pde-example}  so that numerical schemes based on Carleman linearisation are stable.
Recall that for the ODE we have the stability condition $R<1$ with
\begin{equation}
{R = \frac{\norm{F_M}\cdot \uin^{M-1}}{\abs{\lambda_0}} } .
\end{equation}
That condition is not ideal here, because the 2-norm of the solution increases with the number of discretisation points.
Thus this condition for the stability depends not only on the underlying PDE and initial state but on its discretisation.

Ideally we would aim for a condition on the max-norm of the solution.
That can then be used in order to guarantee stability of the solution as well as to bound error.
For example, Ref.~\cite{an2022efficient} considers stability in their Lemma 2.1 and bounds error in their Theorem 3.3.
A simple stability criterion can be given as
\begin{equation}\label{eq:stabPDE}
    \uin_{\max}^{M-1}\frac{b}{|c|} < 1\, .
\end{equation}
Before discretisation, the stability can be shown simply by considering the infinitesimal time interval $\mathrm{d}t$ and using
\begin{align}\label{eq:ustab}
    \| u( \xx, t) + \mathrm{d}t [D \Delta  u(\xx,t) + c  u(\xx,t)+ b  u^{M} (\xx,t)] \|_{\max} &= \| (\identity + \mathrm{d}t \, D \Delta) \{ u( \xx, t) + \mathrm{d}t [c  u(\xx,t)+ b  u^{M} (\xx,t)]\} \|_{\max}  \nn
    &\le \| \identity + \mathrm{d}t \, D \Delta \|_{\infty}
    \| u( \xx, t) + \mathrm{d}t [c  u(\xx,t)+ b  u^{M} (\xx,t)] \|_{\max} .
\end{align}
Now using the triangle inequality
\begin{align}
    \| u( \xx, t) + \mathrm{d}t [c  u(\xx,t)+ b  u^{M} (\xx,t)] \|_{\max} &\le
    \| u( \xx, t) + \mathrm{d}t \, c  u(\xx,t)\|_{\max}  + \mathrm{d}t \| b u^{M} (\xx,t) \|_{\max}\nn
    &= (1+\mathrm{d}t\, c) \| u( \xx, t) \|_{\max} + \mathrm{d}t\, b  \| u (\xx,t) \|_{\max}^{M} \, .
\end{align}
Then if $(b/c) \, \| u (\xx,t) \|_{\max}^{M-1} \le 1$ this expression is upper bounded by $\| u (\xx,t) \|_{\max}$.

Moreover, it is a standard result that $\| \identity + \mathrm{d}t \, D \Delta \|_{\infty}=1$.
That is, the diffusion equation smooths out any peaks in the distribution.
That expression also holds if we consider the discretised form, but only using the first-order discretisation.
Then for spatial grid spacing $h$, the discretised form in one dimension has $1-2\mathrm{d}t\, D /h^2$ on the diagonal, and $\mathrm{d}t\, D /h^2$ on the two off-diagonals.
The sum of the absolute values along a row for this matrix is then exactly 1, giving an $\infty$-norm of 1.
That means Eq.~\eqref{eq:ustab} implies $\| u (\xx,t) \|_{\max}$ is non-increasing for the PDE given the stability criterion in Eq.~\eqref{eq:stabPDE}.

This result for the $\infty$-norm no longer holds for the discretised PDE when using higher-order discretisations.
For example, for the second-order discretisation, the $-1/12$ on the off-diagonals means that 
\begin{equation}
    \| \identity + \mathrm{d}t \, D L_2 \|_{\infty} = 1 + \mathrm{d}t \, D/(3 h^2) \, .
\end{equation}
That means that the max-norm is only upper bounded by the initial max-norm multiplied by a factor of $\exp(D/(3 h^2))$.
In practice it is found that the max-norm is far better behaved.
If we calculate the $\infty$-norm of $\exp(t D L_k)$, then we obtain the results shown in Fig.~\ref{fig:maxnorm}.
For the second-order discretisation the initial slope is $1/3$ as predicted using infinitesimal $t$, but the peak value is less than 1\% above 1.
For the higher-order discretisations this maximum increases, but it is still small for these orders.
Therefore we find that if we consider the max-norm for just evolution under the discretised Laplacian then it is well-behaved, but that does not imply the result for the nonlinear discretised PDE.

\begin{figure}
    \centering
    \includegraphics[width=0.5\linewidth]{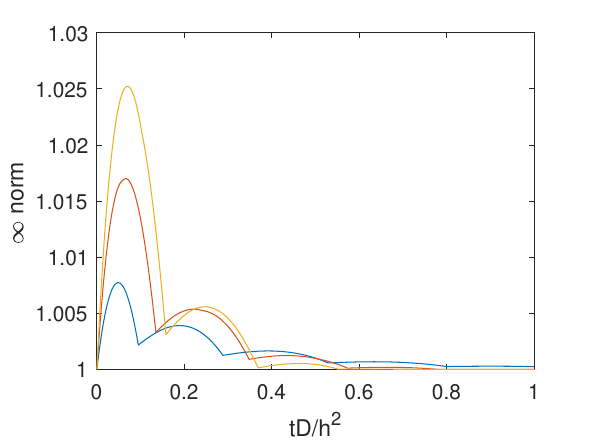}
    \caption{The (induced) $\infty$-norm of the exponential of $tD\Delta_h$ when using a discretisation of the Laplacian of second order (blue), third order (red), and fourth order (orange).}
    \label{fig:maxnorm}
\end{figure}

To determine stability for the 2-norm, the equivalent of Eq.~\eqref{eq:ustab} gives
\begin{align}\label{eq:ustab2}
    \| u( \xx, t+\mathrm{d}t) \| 
    &\le \left\| \identity + \mathrm{d}t \, D \Delta \right\|\times
    \left\| u( \xx, t) + \mathrm{d}t \left[c  u(\xx,t)+ b  u^{M} (\xx,t)\right] \right\| \nn
    &\le \left\| u( \xx, t) + \mathrm{d}t \left[c  u(\xx,t)+ b  u^{M} (\xx,t)\right] \right\|\nn
    &\le \| u( \xx, t)\| + \mathrm{d}t \left[c  \|u(\xx,t)\| + b  \|u^{M} (\xx,t)\|\right]  \nn
    & \le \| u( \xx, t)\| + \mathrm{d}t \left[c  \|u(\xx,t)\| + b  \|u (\xx,t)\|_{\max}^{M-1} \|u (\xx,t)\|\right] \, .
\end{align}
Therefore the 2-norm is non-increasing provided $(b/c) \, \| u (\xx,t) \|_{\max}^{M-1} \le 1$.
In the discretised case the non-positive eigenvalues of the discretised Laplacian mean that the 2-norm is still stable given this condition, though as noted above it is possible for $\| \uu \|_{\max}$ to increase above its initial value.

However, for the purpose of solving the ODE using Carleman linearisation, what matters is not the stability of the nonlinear equation, but that of the linearised equation.
That is because the stability of the linearised equation governs the condition number of the linear equations to solve, and in turn that is proportional to the complexity.
For example, in Ref.~\cite{an2022efficient} their Problem 1 assumes that $\|F_M\|\le |\lambda_0|$ ($\lambda_1$ in the notation of that work) after some possible rescaling of the equation.
Then Eq.~(4.15) of that work gives $\|\identity+Ah\|\le 1$ (with $h$ the time discretisation), using that condition from Problem 1.
That is then used to provide the bound on the norm of $\|L^{-1}\|$ in Eq.~(4.28) of that work, which is then used to give the bound on the condition number proportional to the number of time steps in Eq.~(4.29) in Ref.~\cite{an2022efficient}.

According to our analysis above, the stability of the linearised system will be satisfied provided $\gamma^{M-1}\le |\lambda_0|/\|F_M\|$.
For the discretised PDE here we have $\lambda_0=c$ and $\|F_M\|=b$.
That means if $\gamma\le \|\mathbf{u}_{\rm in}\|_{\max}$, then the stability condition in Eq.~\eqref{eq:stabPDE} implies the stability of the matrix after Carleman linearisation.
That condition is needed in order to be able to use the ODE solver of Ref.~\cite{berry2022quantum}, but it will mean that the rescaling gives a smaller probability of success for obtaining the correct component of the Carleman vector than if we had the stability condition $R<1$.

However, if we have sufficiently small $b/|c|$, then the condition $R<1$ would be satisfied, so
\begin{equation}
    \uin^{M-1}\frac{b}{|c|} < 1\, .
\end{equation}
Because $\uin$ increases with the number of discretisation points as $\sqrt{n}$, this inequality can only be satisfied if the number of discretisation points is made as small as possible.
This gives a strong motivation for using the higher-order spatial discretisation of the PDE.
See \cref{app:trad_off} for discussion of the number of points needed.

\subsection{Error Analysis}

The overall error $\varepsilon$ comes from three different parts,
\begin{itemize}
    \item the spatial discretisation error of the semi-discrete dynamics $\varepsilon_{\mathrm{disc}}$, 
\item the error $\varepsilon_{\mathrm{Carl}}$ contributed by truncation in the Carleman linearisation as bounded in~\cref{lem:error_Comp}, and
\item the error in the time evolution $\varepsilon_{\mathrm{time}}$ due to the Taylor series, as described in~\cref{sec:ODEs}.
\end{itemize}
As usual in this type of analysis, we can simplify the discussion by taking the error to be $\varepsilon$ for each of these contributions.
In reality, the contribution to the error from each source would need to be taken to be a fraction of $\varepsilon$ (e.g.~$\varepsilon/3$), but because that fraction would at most give a constant factor to the complexity, it would not affect the complexities quoted using $\mathcal{O}$.

We have already considered $\varepsilon_{\mathrm{time}}$ and $\varepsilon_{\mathrm{Carl}}$ above in Corollary \ref{cor:pde-complexity}.
The time discretisation error will not be further considered here, but we will discuss how the Carleman error can be alternatively bounded in situations where the PDE is stable but $R\ge 1$.
Above we show that the ODE needs $R<1$ for the quantum solution to be efficient, but this bound on the Carleman error will be useful if that limitation can be circumvented.

In the case of higher-order discretised Laplacians we obtain a somewhat worse bound as derived in \cref{app:carl_errMax},
\begin{equation}
    \|\mathbb{\eta}_{j}(t)\|_{\max} \lesssim G_\kappa^{dj} \left(\frac{\|F_M\|_{\infty}}{|c|}\|\mathbf{u}_{\mathrm{in}}\|_{\max}^{M-1}G_\kappa^{dM}\right)^{k} 
f_{j,k,M}\left( |c|t \right) \, ,
\end{equation}
where
\begin{equation}
    G_\kappa := \max_{\tau\ge 0} \| e^{ L_{\kappa}\tau} \|_{\infty} \, .
\end{equation}
The $\lesssim$ is because it is assuming that the max-norm of the solution is not increasing.
We can use $\le$ if $\|\mathbf{u}_{\mathrm{in}}\|_{\max}$ is replaced with the maximum of $\|\mathbf{u}\|_{\max}$ over time.
The quantity $G_\kappa$ is greater than 1 for higher-order discretised Laplacians, so this is a slightly larger upper bound than in the case of first-order discretised Laplacians.
Nevertheless, the Carleman error may be made arbitrarily small with order provided
\begin{equation}
    \|\mathbf{u}_{\mathrm{in}}\|_{\max}^{M-1}\frac{b}{|c|}G_\kappa^{dM} < 1 \, .
\end{equation}
This condition is slightly stronger than the condition for stability of the PDE by the factor of $G_\kappa^{dM}$, but that will typically be close to 1.
Typically this will be a much weaker requirement than the stability condition $R<1$.

Next, we consider the bound on the error due to the spatial discretisation, which can be used to derive the appropriate number of discretisation points $n$ to use.
Using that in \cref{cor:pde-complexity} then gives the complexity entirely in terms of the parameters of the problem instead of the value chosen for $n$.
Our bound on the error is as given in the following Lemma, with
the proof given in \cref{app:spat-discretisation-semidiscrete}.

\begin{lemma}[Nonlinear PDE solution error  when discretising the Laplacian with higher-order finite differences]\label{lem:spat-discretisation-semidiscrete}
Using a higher-order finite difference discretisation with $2k+1$ stencil points in each direction, the solution of the PDE
\begin{equation}
\frac{\mathrm{d} \uu}{\mathrm{d}t}= (D L_{k,d}  + c) \uu + b \uu^{\otimes M},
\end{equation}
at time $T>0$ has error due to spatial discretisation when $c<0$ and $|c|> M|b|\uin^{M-1}$ bounded as 
\begin{equation}
\label{eq:ineq_error}
\norm{\varepsilon_{\mathrm{disc}}(T)} = \order{   
C(u,k)\sqrt{n}\left(\frac{e}{2}\right)^{2k} n^{- (2k - 1)/d} \frac{1 - \exp{\left( c + M|b|{\uin}_{\max}^{M-1}\right)t} }{ \left|c + M|b|{\uin}_{\max}^{M-1}\right|}},
\end{equation}
where $C(u,k)$ given in \cref{eq:Coef_C} is a constant depending on the $(2k+1)$st spatial derivative of the solution assuming sufficient regularity, $n$ is the number of grid points used, and $d$ is the number of dimensions. 
\end{lemma}

In this result, we are considering the continuous time evolution.
Note that for the stability of the discretisation error we use the condition $|c|> M|b|\uin^{M-1}_{\max}$, which is stronger than the condition $|c|> |b|\uin^{M-1}_{\max}$ for PDE stability.
This appears to be a fundamental condition due to the nonlinearity, because the derivative of the order-$M$ nonlinearity produces a factor of $M$.

Next, if we take $\varepsilon_{\text{disc}}\propto\varepsilon$, then solving Eq.~\eqref{eq:ineq_error} for $n$ gives
\begin{equation}\label{eq:nchoice}
n = \Omega \left(\left[\frac{C(u,k)\left(\frac{e}{2}\right)^{2k}}{\left|c +  M|b|\uin^{M-1}_{\max}\right|}\frac{1}{\varepsilon}\right]^{\frac{2d}{2(2k-1)+d}}\right) .   
\end{equation}
That is, this choice of $n$ is sufficient to give $\varepsilon_{\text{disc}}$ as some set fraction of $\varepsilon$.
In practice we would choose the minimum $n$ needed to give the desired accuracy, so we would choose $n$ proportional to the expression in Eq.~\eqref{eq:nchoice}.
For simplicity of the solution for $n$, we have used
\begin{equation}
 1 - \exp{\left( c + M|b|\uin^{M-1}_{\max}\right)t}\leq 1 \, .  
\end{equation}
In \cref{app:trad_off} we show how the discretisation error is reduced with the number of grid points with different orders of discretisation and a simple toy model.

The benefit of having fewer grid points comes with the drawback of having a less sparse operator.
The block encoding of $F_1$ is performed by a linear combination of unitaries over $(2k+1)$ basis states with amplitudes given by a table, and so has complexity in terms of elementary gates proportional to $k$.
That is not immediately obvious from \cref{cor:pde-complexity}, because it gives complexity in terms of block encodings of $F_1$ and $F_M$.
It is also possible for $C(u,k)$ to increase with $k$.
In a real implementation it would therefore be desirable to choose an optimal $k$ to minimise the complexity, instead of taking $k$ as large as possible.

\section{Conclusion}\label{sec:conclusion}

In this study, we proposed a set of improvements to quantum algorithms for nonlinear differential equations via Carleman linearisation, eliminating some of the exponential scalings seen in prior work.
We have examined both ODEs, and a class of nonlinear PDEs corresponding to reaction-diffusion equations \cite{an2022efficient}.
These improvements include
\begin{itemize}
    \item rescaling the original dynamics,
    \item a truncated Taylor series for the time evolution,
    \item higher-order spatial discretisation of the PDEs, and
    \item tighter bounds on the error of Carleman linearisation.
\end{itemize}
The rescaling boosts the success probability, needing exponentially fewer steps for the amplitude amplification to obtain the solution component of interest in the linearised ODE system.
That is vital to enable the complexity of the PDE solver to be sublinear in the number of discretisation points.
The solution approximation via the truncated Taylor method gives a near-linear dependence on $T$, the total evolution time.
The higher-order spatial discretisation greatly improves the complexity of the quantum solution of PDEs, because it reduces the number of discretisation points needed, which is needed to avoid stability problems.

We show that the stability criterion for PDEs, rescaling, Carleman linearisation, and stability criterion for ODE solvers all interact in a way that makes the solution of PDEs more challenging than was appreciated in prior work. In particular, the stability criterion for PDEs is in terms of a max-norm, but rescaling by the 2-norm is required to obtain a reasonable probability for the correct component of the Carleman vector.
But, rescaling by the 2-norm can make the resulting system of linearised equations unstable,
which causes the ODE solver to have exponential complexity.
If the discretised PDE is still stable in terms of the 2-norm, then the resulting quantum algorithm will still be efficient.
Because the 2-norm increases as $\sqrt n$ in the number of discretisation points, the number of those points should be made as small as possible, which is why it is crucial to use the higher-order discretisation.

In future work, one could devise a less restricted quantum algorithm for solving nonlinear PDEs via some other approach. The feature that the linearised equations can be unstable even though the nonlinear equation is stable suggests that an alternative linear equation solver may be efficient.
The reason why the condition number is large (causing the inefficiency) is that the solution can grow exponentially over time, but for an initial vector that is not of the Carleman form.
A solver that is able to take advantage of the restricted form of states could potentially be efficient.

Furthermore, there are a number of important generalisations that can be made to the type of differential equations. Instead of just including a nonlinear term of order $M$, one could include all nonlinear orders up to $M$.
That could also be used to analyse the effect of driving because the method used for quadratic nonlinearities would produce nonlinearities at a range of orders.
A further generalisation that could be considered is time-dependent differential equations.
These generalisations can be made in a simple way in the quantum algorithm, but the analysis to bound the error would be considerably more complicated.

\section*{Acknowledgements}
MESM and PS were supported by the Sydney Quantum Academy, Sydney, NSW, Australia.
DWB worked on this project under a sponsored research agreement with Google Quantum AI. DWB is also supported by Australian Research Council Discovery Projects DP190102633, DP210101367, and DP220101602.

%


\appendix

\section{Variable names and conventions}
We generally denote scalars using lower case, e.g.\ $c$, vectors using lower-case bold, e.g.\ $\uu$, and matrices/operators using upper-case letters, such as $A$. Unless noted otherwise, norms correspond to the spectral or Euclidean 2-norm. The notation $[N]$ describes the set $\{1, 2, \ldots, N\}$.  
\begin{itemize}
    \item $\xx$ -- Coordinate in space. 
    \item $t$ -- Coordinate in time.
    \item $T$ -- Final time.
    \item $\uu$ -- Solution vector for nonlinear ODE.
    \item $\yy$ -- Solution vector of system arising from Carleman linearisation, $\yy = [\uu, \uu^{\otimes 2}, \ldots, \uu^{\otimes N}, \ldots]^N$.
    \item $F_1, F_M$ -- Terms appearing in nonlinear ODE in \cref{defn:nonlin-ode-problem}, where $F_1$ is an operator representing a linear term and $F_M$ corresponds to a nonlinearity of the order of the subscript $M$. 
    \item $N$ -- Truncation number of Carleman linearisation.
    \item $\mathcal{A}, \mathcal{A}_N$ -- System matrix after Carleman linearisation, subscript $N$ denotes truncated system; see \cref{eq:carlmatrix}.
    \item $A_j^{(i)}$ -- Components of block-structure of Carleman matrix ($j$th row, $(i-1)$th diagonal). Defined in \cref{eq:carleman-assembly}.
    \item $\varepsilon$ -- Solution error. 
    \item $\eta$ -- Error in Carleman vector, defined in \cref{lem:error_Glob,lem:error_Comp}.
    \item $\Delta$ -- Laplacian operator.
    \item $L$, $L_k$ -- Finite difference approximation of the Laplacian, up to ``order'' $k$, so that there is a central finite difference stencil with $2k+1$ points.
    \item $D, c, b$ -- Coefficients from PDE in \cref{eq:pde-example}; $D$ - diffusion, $c$ decay, $b$ nonlinearity. 
    \item $n$ -- Total number of gridpoints in the discretisation of the PDE example in \cref{eq:pde-example}.
    \item $d$ -- Dimensionality of the PDE problem in \cref{eq:pde-example}.
    \item $\identity$ -- Identity matrix, indexing may clarify the dimension rather than number of qubits.
    \item $\lambda_j(\cdot)$ -- Eigenvalue $j$ of matrix $(\cdot)$, whereas $\lambda_{(\cdot)}$ is the subnormalization factor of the block-encoding of $(\cdot)$.
    \item $\gamma$ -- Rescaling factor introduced in \cref{defn:rescaled-carleman}, $\gamma>0$. 
    \item $\widetilde{(\cdot)}$ -- Rescaled quantities as defined in \cref{defn:rescaled-carleman}.
    \item $R$ -- Ratio of strength of nonlinearity over decay, \cref{eq:R}.
    \item $f_{j,k,M}$ -- Function to express tighter bound for Carlemann errors, defined in \cref{lem:error_Comp}.
    \item $\Omega_k$ -- Carleman error intervals, used in \cref{lem:error_Comp}.
    \item $\eta$ -- Error due to Carleman linearisation, defined in \cref{lem:error_Glob}. 
    \item $s(\cdot)$ -- Sparsity of quantity $(\cdot)$. 
\end{itemize}

\section{Proofs of bounds on error due to Carleman linearisation}\label{app:proofs_bounds}
\subsection{Proof of bound on the full vector of errors due to Carleman linearisation}
\label{app:proof_total_error}
Here, we present the proof of \cref{lem:error_Glob}.
\begin{proof}
    We define the error due to Carleman truncation of order $N$, for any $j\in [N]$, including a rescaling by $\gamma>0$, as
    \begin{align}
    \label{eq:error_car_def}
    \mathbf{\eta}_j(t) &\coloneqq \widetilde{\uu}^{\otimes j} - \widetilde{\mathbf{y}}_j\nn 
    &=\frac{\mathbf{u}^{\otimes j} - \mathbf{y}_j}{\gamma^j} .    
    \end{align}
    Here, $\mathbf{u}$ is the exact solution of the nonlinear ODE, and $\mathbf{y}_j$ are the components of the solution vector using the truncated Carleman approximation.
    Therefore, $\mathbf{\eta}_j$ corresponds to the error arising from the Carlemann linearisation at finite order, but with the rescaling $\gamma$.
    This definition is similar to that in \cite{liu2020efficient} (see the proof of Lemma 1 in the Supplementary Information of that work), except we are including the rescaling.
    By writing $\mathbf{U}$ as a vector that has the components $\widetilde{\mathbf{U}}_j=\widetilde{\uu}^{\otimes j}$,
    we have from \cref{eq:error_car_def} that
    \begin{align}
    \label{eq:two_comp}
    \dt{\eta_j(t)} &=  \left(\widetilde{\mathcal{A}} \widetilde{\mathbf{U}}\right)_j -  \left(\widetilde{\mathcal{A}}_N\widetilde{\mathbf{y}}\right)_j    \nn
     &=  A_j^{(1)}\widetilde{\uu}^{\otimes j} +  \widetilde{A}_{j+M-1}^{(M)}\widetilde{\uu}^{\otimes j +M -1} - A_j^{(1)}\widetilde{\yy}_{ j} -  \widetilde{A}_{j+M-1}^{(M)}\widetilde{\yy}_{ j +M -1}.
    \end{align}
    Notice that in the equation above we have $\widetilde{\mathcal{A}}$ as the rescaled form of $\mathcal{A}$, the infinite-dimensional matrix generated from the Carleman linearisation, as explained in \cref{sec:backgound}, and $\widetilde{\mathcal{A}}_N = \gamma^{M-1} \mathcal{A}_N$ as the truncated matrix as given in \cref{eq:resccarlmatrix}. 
    Because we are taking the difference between the action of an infinite-dimensional matrix $\widetilde{\mathcal{A}}$ and a truncated one $\widetilde{\mathcal{A}}_N$, we need different treatment for values of $j$ below and above that truncation.
    For $j\le N-M+1$ we are below the truncation, and in \cref{eq:two_comp} we can use $\mathbf{\eta}_j(t) = \widetilde{\uu}^{\otimes j} - \widetilde{\mathbf{y}}_j$ to give
    \begin{equation}
    \dt{\eta_j} =  A_j^{(1)}\eta_j +  \widetilde{A}_{j+M-1}^{(M)}\eta_{j+M-1} .
    \end{equation}
    Then, using $\widetilde{\mathcal{A}}_j = \gamma^{M-1} \mathcal{A}_j$ gives
    \begin{equation}
    \dt{\eta_j} =  A_j^{(1)}\eta_j +  \gamma^{M-1}A_{j+M-1}^{(M)}\eta_{j+M-1}.
    \end{equation}
 Now, for $j>N-M+1$, we have $j+M-1>N$, and so $\widetilde{\yy}_{ j +M -1}$ is past the end of the vector and should be taken to be zero.
This implies that we obtain
\begin{equation}
    \dt{\eta_j} =  A_j^{(1)}\eta_j +  \widetilde{A}_{j+M-1}^{(M)}\widetilde{\uu}_{j+M-1} .
\end{equation}
Now, using $\widetilde{\mathcal{A}}_j = \gamma^{M-1} \mathcal{A}_j$ and $\widetilde{\uu}_{j+M-1}=\gamma^{-(j+M-1)}\uu_{j+M-1}$ we get for $j>N-M+1$
\begin{align}
    \dt{ \eta_{j}} &=A_{j}^{(1)}\eta_{j} +   \gamma^{M-1} A_{j+M-1}^{(M)}\gamma^{-(j+M-1)}\mathbf{u}^{\otimes(j+M-1)} \nn
    &= A_{j}^{(1)}\eta_{j} +   \gamma^{-j}A_{j+M-1}^{(M)}\mathbf{u}^{\otimes(j+M-1)} .
\end{align}
Therefore, we end up with
    
    \begin{align}
    \label{eq:elem_mat}
        \dt{\eta_j} &=  A_j^{(1)}\eta_j +  \gamma^{M-1}A_{j+M-1}^{(M)}\eta_{j+M-1},\quad j\in[N-M+1] \nn
        \dt{ \eta_{j}} &=  A_{j}^{(1)}\eta_{j} +   \gamma^{-j}A_{j+M-1}^{(M)}\mathbf{u}^{\otimes(j+M-1)}, \quad j\in\{N-M+2, \ldots, N\}, 
    \end{align}
   
    Notice that in the final form above, we left the equations without rescaling notation to show the explicit dependence of $\gamma$ in the error equation. In a more compact form, the equation above can be written as
    \begin{equation}\label{eq:carleman_ODE}
    \dt{ \mathbf{\eta}} = \widetilde{\mathcal{A}}_N \mathbf{\eta} + \mathbf{b},  
    \end{equation}
    which has the following block matrix structure
    \begin{equation}
    \label{eq:Matrix_Carl_rescaled}
    \dt{}\begin{bmatrix}\mathbf{\eta}_1\\ \mathbf{\eta}_2 \\ \vdots\\ \mathbf{\eta}_{N-1} \\ \mathbf{\eta}_{N} \end{bmatrix} 
    =
    \underbrace{
   \begin{bmatrix}
    A_1^{(1)} & 0  & \cdots &  0 & 0 &  \gamma^{M-1}A^{(M)}_{M} & 0  & 0  \\
    0 &  A_2^{(1)}  & \cdots& 0 & 0 & 0 &  \gamma^{M-1}A^{(M)}_{M+1}  & 0 \\
    \vdots &   &   & \vdots & \vdots & &  \cdots  &  \vdots\\
    0 & 0 &  \cdots & 0 & 0 & 0 & \ A_{N-1}^{(1)}  & 0\\
    0 & 0 & \cdots & 0 & 0 & 0 & 0 &  A_N^{(1)} \\
    \end{bmatrix}.
    }_{\widetilde{\mathcal{A}}_N}
    \begin{bmatrix} \mathbf{\eta}_1\\ \mathbf{\eta}_2  \\ \vdots\\ \mathbf{\eta}_{N-1} \\ \mathbf{\eta}_{N} \end{bmatrix} 
    +
    \underbrace{
    \begin{bmatrix} \mathbf{b}_1\\ \mathbf{b}_2  \\ \vdots\\ \mathbf{b}_{N-1} \\ \mathbf{b}_{N} \end{bmatrix}   
    }_{\mathbf b},
    \end{equation} 
    where  
    \begin{align}
    \label{eq:compb}
    \mathbf{b}_j &= \mathbf{0}^{\otimes j},\quad j\in[N-M+1] \nn 
    \mathbf{b}_j &=  \gamma^{-j}A_{j+M-1}^{(M)}\mathbf{u}^{\otimes(j+M-1)} \quad j\in\{N-M+2, \ldots, N\}.
     \end{align}
    and $ \mathbf{0}$ represents a vector with all entries equal to zero with the same dimension as $\mathbf{u}$. The blocks $A_j^{(1)}$ and  $A^{(M)}_{M+j-1}$ are given in \cref{eq:carleman-assembly}.
    Similarly to Eq.\ (35) in the Supplementary Information of 
     \cite{liu2020efficient}, we obtain by considering the  derivative of $\|\mathbf{\eta}\|$ that
    \begin{align}
    \label{eq:deriv_norm}
    \dt{}(\mathbf{\eta}^{\dagger} \mathbf{\eta}) &= \dt{\mathbf{\eta}^{\dagger}}\mathbf{\eta} + \mathbf{\eta}^{\dagger}\dt{\mathbf{\eta}} \nn
    & = \mathbf{\eta}^{\dagger}\left( \widetilde{\mathcal{A}}_N +\widetilde{\mathcal{A}}^{\dagger}_N\right)\mathbf{\eta} + \mathbf{\eta}^{\dagger}\mathbf{b} +\mathbf{b}^{\dagger}\mathbf{\eta}.
    \end{align}

Using Gershgorin's circle theorem~\cite{Horn2012,Feingold1962BlockDD,VANDERSLUIS1979265}, the maximum eigenvalue of $\widetilde{\mathcal{A}}_N +\widetilde{\mathcal{A}}^{\dagger}_N$ is at most (see the discussion above for Eq.~\eqref{eq:matcirc})
    \begin{equation}
        2j\lambda_0 + (2j-M+1)\gamma^{M-1}\|F_M\| \, .
    \end{equation}
     We then have,
    \begin{equation}
    \label{eq:boundMat}
    \mathbf{\eta}^{\dagger}\left( \widetilde{\mathcal{A}}_N +\widetilde{\mathcal{A}}^{\dagger}_N\right)\mathbf{\eta}\leq \left(2N\lambda_0+(2N-M+1) \gamma^{M-1}\|F_M\|\right)\|\mathbf{\eta}\|^2.
    \end{equation}
    The remaining term to bound in Eq.~\eqref{eq:deriv_norm} is $\mathbf{\eta}^{\dagger}\mathbf{b} +\mathbf{b}^{\dagger}\mathbf{\eta}$, which has the upper bound
    \begin{align}
    \eta^{\dagger}\bb +\bb^{\dagger}\eta &\leq 2\|\bb\|\|\eta\|\nn
    &= 2\sum_{j=N-M+2}^{N} \gamma^{-j}\|A^{(M)}_{j+M-1}\| \cdot \|\mathbf{u}^{\otimes(j+M-1)}\| \cdot \|\eta\| \nn
    &= 2\|F_M\|\sum_{j=N-M+2}^{N} j \gamma^{-j}\|\uu^{\otimes(j+M-1)}\| \cdot \|\mathbf{\eta}\| \, ,
    \end{align}
    where we used \cref{eq:compb} for the vector norm of $\bb$, followed by using \cref{eq:carleman-assembly}.
    Now by taking the maximum value that $j$ can assume in the summation above gives
    \begin{align}
    \label{eq:vec_part}
    \eta^{\dagger}\bb +\bb^{\dagger}\eta &\leq 2N\|F_M\|\sum_{j=N-M+2}^{N}  \gamma^{-j} \|\uu^{\otimes(j+M-1)}\| \cdot \|\eta\| .
    \end{align}
    As stated in the theorem, we are dealing with dissipative problems so $\|\mathbf{u}(t)\| \leq \|\uu_{\mathrm{in}}\|$ for $t>0$, and we also take $\gamma = \|\uu_{\mathrm{in}}\|$.
    Therefore
    \begin{equation}\label{eq:vec_part2}
        \eta^{\dagger}\bb +\bb^{\dagger}\eta \leq 2N(M-1)\|F_M\|\|\uu_{\mathrm{in}}\|^{M-1} \|\mathbf{\eta}\| .
    \end{equation}
    
    Returning to \cref{eq:deriv_norm} together with the fact that $\dt{ \|\mathbf{\eta}\|} = (2\|\mathbf{\eta}\|)^{-1} \dt{\|\mathbf{\eta}\|^2}$, we have
    \begin{equation}
    \label{eq:dterror}
    \dt{\|\mathbf{\eta}\|} = (2\|\mathbf{\eta}\|)^{-1}\left[\mathbf{\eta}^{\dagger}\left( \mathcal{A}_N +\mathcal{A}^{\dagger}_N\right)\mathbf{\eta} + \mathbf{\eta}^{\dagger}\mathbf{b} +\mathbf{b}^{\dagger}\mathbf{\eta}\right]. 
    \end{equation}
    We can now combine the results from \cref{eq:boundMat} and \cref{eq:vec_part2} in the equation above to get 
    \begin{equation}
    \dt{\|\mathbf{\eta}\|} \leq N\left(\lambda_0 + \gamma^{M-1}\|F_M\|\right)\|\mathbf{\eta}\| + N(M-1)\|F_M\| \|\uu_{\mathrm{in}}\|^{M-1}.
    \end{equation}
    
    We know that if we have a differential equation given by
    \begin{equation}
     \dt{\|\mathbf{\eta}\|} = A\|\mathbf{\eta}\| + B,   
    \end{equation}
    where  $\|\mathbf{\eta}(0)\|= 0$, its solution is given by  
    \begin{equation}
    \label{eq:pde_sol}
    \|\mathbf{\eta}(t)\| = \frac{B}{A}\left(1 - e^{At}\right).
    \end{equation}
    Since in our case we have 
    \begin{align}
    A&=N\left(\lambda_0+\gamma^{M-1}\|F_M\|\right) \\
    B&=N(M-1)\|F_M\| \|\uu_{\mathrm{in}}\|^{M-1},\nonumber
    \end{align}
    we get
    \begin{equation}
    \|\mathbf{\eta}(t)\| \leq (M-1)\|F_M\|\ \|\uu_{\mathrm{in}}\|^{M-1}\frac{1-e^{N(\lambda_0 + \gamma^{M-1}\|F_M\|)t}}{|\lambda_0+\gamma^{M-1}\|F_M\||} \, .
    \end{equation}
    This concludes the proof as $\|\mathbf{\eta}(t)\|\geq \|\mathbf{\eta}_j(t)\|$ for all $t$.
\end{proof}

\subsection{Proof of bound on components of the vector of errors due to Carleman linearisation}
\label{app:proof_comp}

In this section we present the proof to \cref{lem:error_Comp}.
For this proof we use the definition of $f_{j,k,M}(|\lambda_0|t)$ as
\begin{align}\label{eq:fdef}
    f_{j,k,M}(\tau) &:= \left[\prod_{\ell=0}^{k-1}(\ell M-\ell+j)\right] \int_0^\tau d\tau_1 \, e^{-j(\tau-\tau_1)}\int_0^{\tau_1} d\tau_2 \, e^{-[(M-1)+j](\tau_1-\tau_2)}
    \int_0^{\tau_2} d\tau_3 \, e^{-[2(M-1)+j](\tau_2-\tau_3)} \cdots \nn & \qquad
    \cdots \int_0^{\tau_{k-1}} d\tau_k \, e^{-[(k-1)(M-1)+j](\tau_{k-1}-\tau_k)} \, ,
\end{align}
This function has the closed-form expression given in \cref{lem:error_Comp},
which is proven in \cref{thm:fvals} below.
We plot these formulae for the example $M=2$ in Fig.~\ref{fig:factor}.
The blue curve on the left is the simple upper bound, so we can see that our expression gives significantly tighter bounds.

\begin{figure}
\begin{center}
	\includegraphics[scale=0.5]{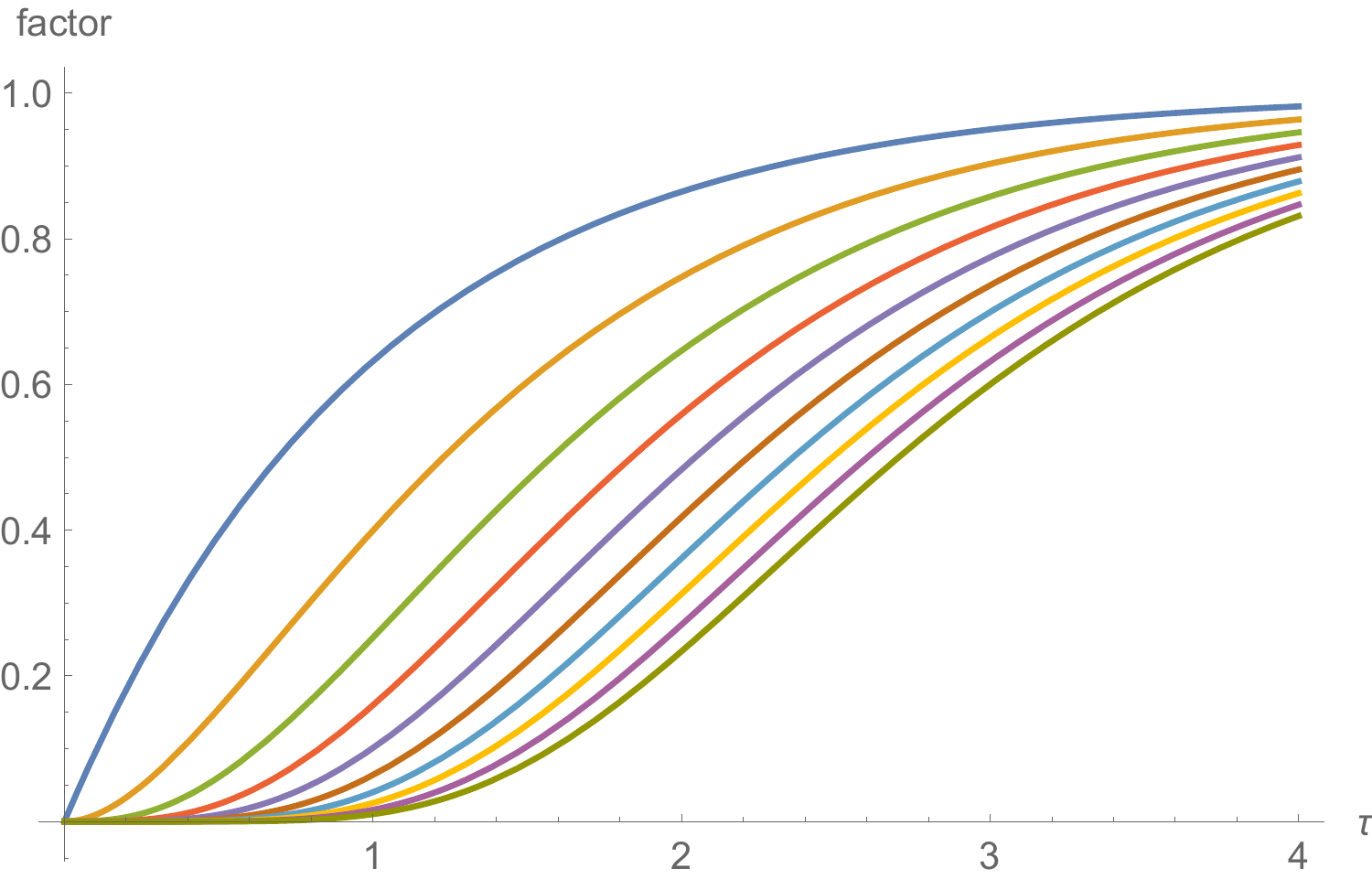}
	\caption{The factors in the upper bound $f_{j,k,M}(\tau)$ with $j=1$ and $M=2$.
The lines are for $k$ equal to $1$ to $10$ ordered from left to right.
As we show, $f_{j,k,M}(\tau)$ is monotonically decreasing with $k$.
}\label{fig:factor}
\end{center}
\end{figure}

\begin{theorem}\label{thm:fvals}
    The function $f_{j,k,M}(\tau)$ defined by Eq.~\eqref{eq:fdef}
for real $\tau\ge 0$ and natural numbers $j\ge 1$, $k\ge 1$, and $M\ge 2$, satisfies
\begin{equation}
    f_{j,k,M}(\tau) = 1-\frac{(M-1) \, \Gamma(k+j/(M-1))}{(k-1)! \, \Gamma(j/(M-1))}\sum_{\ell=0}^{k-1} (-1)^\ell \binom{k-1}{\ell} \frac{e^{-(\ell M-\ell+j) \tau}}{\ell M-\ell+j} \, .
\end{equation}
\end{theorem}
\begin{proof}
We first regroup the exponentials as
\begin{align}
    f_{j,k,M}(\tau) &= \left[\prod_{\ell=0}^{k-1}(\ell M-\ell+j)\right] e^{-j\tau} \int_0^\tau d\tau_1 \, e^{-(M-1)\tau_1}\int_0^{\tau_1} d\tau_2 \, e^{-(M-1)\tau_2}
    \int_0^{\tau_2} d\tau_3 \, e^{-(M-1)\tau_3} \cdots \nn & \qquad
    \cdots \int_0^{\tau_{k-1}} d\tau_k \, e^{[(k-1)(M-1)+j]\tau_n} \, .
\end{align}
The integral is over values of $\tau_\ell$ satisfying $\tau\ge \tau_1\ge\tau_2 \ge \cdots \ge \tau_k$.
We can exchange the order of the integrals in the standard way so they are given as
\begin{align}
    f_{j,k,M}(\tau) &= \left[\prod_{\ell=0}^{k-1}(\ell M-\ell+j)\right] e^{-j\tau} \int_{0}^\tau d\tau_k \, e^{[(k-1)(M-1)+j]\tau_k}
    \int_{\tau_k}^{\tau} d\tau_{k-1} \, e^{-(M-1)\tau_{k-1}} \cdots \nn & \qquad
    \cdots 
    \int_{\tau_3}^{\tau} d\tau_2 \, e^{-(M-1)\tau_2} \int_{\tau_2}^\tau d\tau_1 \, e^{-(M-1)\tau_1} \, .
\end{align}
The advantage of this form is that we have a sequence of integrals of the same form, so we can give a general expression that can be proven by induction.
In particular, let us define the repeated integral
\begin{align}
    g_{k,M}(\tau_k,\tau) &:= 
    \int_{\tau_k}^{\tau} d\tau_{k-1} \, e^{-(M-1)\tau_{k-1}} 
    \cdots 
    \int_{\tau_3}^{\tau} d\tau_2 \, e^{-(M-1)\tau_2} \int_{\tau_2}^\tau d\tau_1 \, e^{-(M-1)\tau_1} \, .
\end{align}
The closed-form expression will be shown to be, for $k\ge 1$,
\begin{equation}\label{eq:gval}
    g_{k,M}(\tau_k,\tau) = \frac{ [e^{-(M-1)\tau_k} - e^{-(M-1)\tau}]^{k-1}}{
(k-1)! (M-1)^{k-1}} \, .
\end{equation}
In the case of $k=1$ the definition of $g_{k,M}(\tau_k,\tau)$ can be just taken to be $1$ (no integrals), and Eq.~\eqref{eq:gval} holds.
In the simple case of $k=2$ we have, from the definition
\begin{align}
    g_{2,M}(\tau_2,\tau) &= \int_{\tau_2}^\tau d\tau_1 \, e^{-(M-1)\tau_1} \nn
    &= \frac 1{M-1} \left[ e^{-(M-1)\tau_2} - e^{-(M-1)\tau} \right] \, .
\end{align}
This also agrees with the claimed expression.
More generally, if Eq.~\eqref{eq:gval} holds for $k$, then we obtain
\begin{align}
    g_{k+1,M}(\tau_{k+1},\tau) & = \int_{\tau_{k+1}}^\tau d\tau_k \, e^{-(M-1)\tau_k} g_{k,M}(\tau_k,\tau) \nn
    &= \int_{\tau_{k+1}}^\tau d\tau_k \, e^{-(M-1)\tau_k}
    \frac{ [e^{-(M-1)\tau_k} - e^{-(M-1)\tau}]^{k-1}}{
(k-1)! (M-1)^{k-1}} \nn
&= \left[ -\frac{ [e^{-(M-1)\tau_k} - e^{-(M-1)\tau}]^{k}}{
k! (M-1)^{k}} \right]_{\tau_{k+1}}^\tau \nn
&= \frac{ [e^{-(M-1)\tau_{k+1}} - e^{-(M-1)\tau}]^{k}}{
k! (M-1)^{k}} \, .
\end{align}
Thus Eq.~\eqref{eq:gval} holds for $k+1$, and must hold for all $k\ge 1$ by induction.

Now to obtain the expression for $f_{j,k,M}(\tau)$, we use
\begin{equation}
    f_{j,k,M}(\tau) = \left[\prod_{\ell=0}^{k-1}(\ell M-\ell+j)\right] e^{-j\tau} \int_0^\tau 
    d\tau_k \, e^{[(k-1)(M-1)+j]\tau_n}
    g_{k,M}(\tau_k,\tau) \, .
\end{equation}
Expanding the expression for $g_{k,M}(\tau_k,\tau)$ involving a power into a sum gives
\begin{align}
    f_{n,M}(\tau) &= \frac{\prod_{\ell=0}^{k-1}(\ell M-\ell+j) }{(k-1)! (M-1)^{k-1}}e^{-j\tau} \int_0^\tau 
    d\tau_k \, e^{[(k-1)(M-1)+j]\tau_k}
    \sum_{\ell=0}^{k-1}(-1)^\ell\binom{k-1}{\ell} 
    e^{-(k-1-\ell)(M-1)\tau_{k}} e^{-\ell(M-1)\tau} \nn
    &=\frac{\prod_{\ell=0}^{k-1}(\ell M-\ell+j) }{(k-1)! (M-1)^{k-1}} \int_0^\tau 
    d\tau_n \, 
    \sum_{\ell=0}^{k-1}(-1)^\ell\binom{k-1}{\ell} 
    e^{-(\ell M-\ell+j)(\tau-\tau_k)} \nn
    &=\frac{\prod_{\ell=0}^{k-1}(\ell M-\ell+j) }{(k-1)! (M-1)^{k-1}} \, 
    \sum_{\ell=0}^{k-1}(-1)^\ell\binom{k-1}{\ell} 
    \left[ -\frac{e^{-(\ell M-\ell+j)(\tau-\tau_k)}}{(\ell M-\ell+j)}\right]_{0}^\tau \nn
    &=\frac{\prod_{\ell=0}^{k-1}(\ell M-\ell+j) }{(k-1)! (M-1)^{k-1}} \, 
    \sum_{\ell=0}^{k-1}(-1)^\ell\binom{k-1}{\ell} 
     \frac{1-e^{-(\ell M-\ell+j)\tau}}{\ell M-\ell+j}
    \, .
\end{align}
Next we aim to show that
\begin{equation}\label{eq:gammas}
    \frac{\prod_{\ell=0}^{k-1}(\ell M-\ell+j) }{(M-1)^{k}} = \frac{\Gamma(k+j/(M-1))}{\Gamma(j/(M-1))} \, .
\end{equation}
First, it is obvious that 
\begin{equation}
    \frac{\prod_{\ell=0}^{k-1}(\ell M-\ell+j) }{(M-1)^{k}} = \prod_{\ell=0}^{k-1}(\ell+j/(M-1)) \, .
\end{equation}
Now for $k=0$ both sides of Eq.~\eqref{eq:gammas} are equal to $1$, because there are no factors in the product, and $k+j/(M-1)=j/(M-1)$ so the result is trivial.
If Eq.~\eqref{eq:gammas} is correct for some $k$, we obtain
\begin{align}
    \frac{\Gamma(k+1+j/(M-1))}{\Gamma(j/(M-1))} &= \frac{\Gamma(k+1+j/(M-1))}{\Gamma(k+j/(M-1))}
    \frac{\Gamma(k+j/(M-1))}{\Gamma(j/(M-1))} \nn
    &= (k+j/(M-1)) \prod_{\ell=0}^{k-1}(\ell+j/(M-1)) \nn
    &= \prod_{\ell=0}^{k}(\ell+j/(M-1)) \, .
\end{align}
In the second line we have used $\Gamma(x+1)=x\Gamma(x)$ and assumed Eq.~\eqref{eq:gammas} holds for $k$.
That means Eq.~\eqref{eq:gammas} holds for all $k\ge 0$ by induction.
Thus we can replace the product with the expression in terms of gamma functions to show that
\begin{equation}
    f_{j,k,M}(\tau) = \frac{(M-1) \, \Gamma(k+j/(M-1))}{(k-1)! \, \Gamma(j/(M-1))}\sum_{\ell=0}^{k-1} (-1)^\ell \binom{k-1}{\ell} \frac{1-e^{-(\ell M-\ell+j) \tau}}{\ell M-\ell+j} \, .
\end{equation}

Next we show that $f_{j,k,M}(\tau)$ is equal to $1$ in the limit $\tau\to\infty$.
If we use shifted variables of integration $\tilde\tau_\ell=\tau_\ell-\tau$,
then the definition of $f_{j,k,M}(\tau)$ gives
\begin{align}\label{eq:fshifted}
    f_{j,k,M}(\tau) &= \left[\prod_{\ell=0}^{k-1}(\ell M-\ell+j)\right] \int_{-\tau}^0 d\tilde\tau_1 \, e^{j\tilde\tau_1}\int_{-\tau}^{\tilde\tau_1} d\tilde\tau_2 \, e^{-(M-1+j)(\tilde\tau_1-\tilde\tau_2)}
    \int_{-\tau}^{\tilde\tau_2} d\tilde\tau_3 \, e^{-[2(M-1)+j](\tilde\tau_2-\tilde\tau_3)} \cdots \nn & \qquad
    \cdots \int_{-\tau}^{\tilde\tau_{n-1}} d\tilde\tau_k \, e^{-[(k-1)(M-1)+j](\tilde\tau_{k-1}-\tilde\tau_k)} \, .
\end{align}
The limit $\tau\to\infty$ then gives
\begin{align}
    \lim_{\tau\to\infty}f_{j,k,M}(\tau) &= \left[\prod_{\ell=0}^{k-1}(\ell M-\ell+j)\right] \int_{-\infty}^0 d\tilde\tau_1 \, e^{j\tilde\tau_1}\int_{-\infty}^{\tilde\tau_1} d\tilde\tau_2 \, e^{-(M-1+j)(\tilde\tau_1-\tilde\tau_2)}
    \int_{-\infty}^{\tilde\tau_2} d\tilde\tau_3 \, e^{-[2(M-1)+j](\tilde\tau_2-\tilde\tau_3)} \cdots \nn & \qquad
    \cdots \int_{-\infty}^{\tilde\tau_{k-1}} d\tilde\tau_k \, e^{-[(k-1)(M-1)+1](\tilde\tau_{k-1}-\tilde\tau_k)} \, .
\end{align}
Then it is easy to see that each successive integral gives division by one of the factors in the product, so the overall expression is equal to 1.

Therefore, since the negative exponentials in our expression for $f_{j,k,M}(\tau)$ approach zero in the limit $\tau\to\infty$, the constant term must be equal to $1$, and we can give $f_{j,k,M}(\tau)$ as
\begin{align}
    f_{j,k,M}(\tau) &= 1-\frac{(M-1)\Gamma(k+j/(M-1))}{(k-1)! \, \Gamma(j/(M-1))}\sum_{\ell=0}^{k-1} (-1)^\ell \binom{k-1}{\ell} \frac{e^{-(\ell M-\ell+j) \tau}}{\ell M-\ell+j}
    \, .
\end{align}
\end{proof}

For the use in our derivations, note that $f_{j,k,M}(\tau)$ satisfies the recurrence relation (which is obvious from its definition)
\begin{equation}\label{eq:fiter}
    f_{j,k,M}(\tau) = j \int_0^\tau d\tau_1 \, e^{-j(\tau-\tau_1)}f_{M-1+j,k-1,M}(\tau_1) \, .
\end{equation}
Another useful property is that $f_{j,k,M}(\tau)\le f_{j,k-1,M}(\tau)$, so $f_{j,k,M}(\tau)$ is monotonically decreasing with $k$.
This can be seen using Eq.~\eqref{eq:fshifted}, and upper bounding it by taking the limit as the last integral goes to infinity as
\begin{align}
    f_{j,k,M}(\tau) &\le \left[\prod_{\ell=0}^{k-1}(\ell M-\ell+j)\right] \int_{-\tau}^0 d\tilde\tau_1 \, e^{j\tilde\tau_1}\int_{-\tau}^{\tilde\tau_1} d\tilde\tau_2 \, e^{-(M-1+j)(\tilde\tau_1-\tilde\tau_2)}
    \int_{-\tau}^{\tilde\tau_2} d\tilde\tau_3 \, e^{-[2(M-1)+j](\tilde\tau_2-\tilde\tau_3)} \cdots \nn & \qquad
    \cdots \int_{-\infty}^{\tilde\tau_{k-1}} d\tilde\tau_k \, e^{-[(k-1)(M-1)+j](\tilde\tau_{k-1}-\tilde\tau_k)}  \nn
    &\le \left[\prod_{\ell=0}^{k-2}(\ell M-\ell+j)\right] \int_{-\tau}^0 d\tilde\tau_1 \, e^{j\tilde\tau_1}\int_{-\tau}^{\tilde\tau_1} d\tilde\tau_2 \, e^{-(M-1+j)(\tilde\tau_1-\tilde\tau_2)}
    \int_{-\tau}^{\tilde\tau_2} d\tilde\tau_3 \, e^{-[2(M-1)+j](\tilde\tau_2-\tilde\tau_3)} \cdots \nn & \qquad
    \cdots \int_{-\tau}^{\tilde\tau_{k-2}} d\tilde\tau_{k-1} \, e^{-[(k-2)(M-1)+j](\tilde\tau_{k-2}-\tilde\tau_{k-1})} \nn
    &= f_{j,k-1,M}(\tau)
    \, .
\end{align}

Now we proceed to the proof of \cref{lem:error_Comp}.
\begin{proof}
 We recall from  \cref{eq:elem_mat}  
\begin{align}\label{eq:line1}
\dt{\eta_j} &=  A_j^{(1)}\eta_j +  \gamma^{M-1}A_{j+M-1}^{(M)}\eta_{j+M-1},\quad j\in[N-M+1] \\ 
\dt{ \eta_{j}} &=  A_{j}^{(1)}\eta_{j} +   \gamma^{-j}A_{j+M-1}^{(M)}\mathbf{u}^{\otimes(j+M-1)}, \quad j\in\{N-M+2, \ldots, N\}. \label{eq:line2}
\end{align}

The method to bound the norms of $\eta_{j}$ is to first use the second line \eqref{eq:line2} for $j\in\Omega_1$ with
\begin{equation}
\Omega_1 \coloneqq \{N-M+2,\dots,N\}.
\end{equation}
These bounds can be derived just using that value of $j$, since the equation does not depend on $\eta_{j}$ for any other values of $j$.
For smaller values of $j$ we need to use the first line \eqref{eq:line1}, but that depends on $\eta_{j+M-1}$.
If we have $j\in \Omega_2$ for
\begin{equation}
\Omega_2 \coloneqq \{N-2M+3,\cdots, N-M+1\},    
\end{equation}
then $j+M-1\in \Omega_1$.
We can therefore use the bound derived for $j\in\Omega_1$ to bound the second term in Eq.~\eqref{eq:line1} and thereby derive the bound on $\eta_{j}$ for $j\in\Omega_1$.

Then for $j$ less than $N-2M+3$ we can use the bounds on $\eta_{j}$ for $j\in\Omega_2$, and so forth, to eventually derive a bound on $\eta_{1}$.
In particular we define
\begin{equation}
 \Omega_k \coloneqq \{N-k(M-1)+1,\ldots, N-(k-1)(M-1)\},    
\end{equation}
for $k =1, 2,\cdots, \lceil\frac{N}{M-1}\rceil$.
Then we work backwards in steps of $M-1$ to derive bounds on $\eta_{j}$ for $j\in\Omega_k$ using the bound on $\eta_{j+M-1}$ for $j+M-1\in\Omega_{k-1}$.
In particular, $1\in\Omega_k$ for $k=\lceil\frac{N}{M-1}\rceil$.
This can be seen from the extremal cases; first, if $N$ is a multiple of $M-1$, then $k=N/(M-1)$, so
\begin{equation}
    N-k(M-1)+1 = N - [N/(M-1)](M-1)+1 = 1 \, .
\end{equation}
The other extremal case is where $N-1$ is a multiple of $M-1$, so the ceiling function gives the maximal rounding up of $N/(M-1)$.
Then $k-1=(N-1)/(M-1)$, and
\begin{equation}
    N-(k-1)(M-1) = N - [(N-1)/(M-1)](M-1) = 1 \, .
\end{equation}
By deriving bounds on $\eta_{j}$ in the sequence of $\Omega_k$ sets for $k$ from 1 to $\lceil\frac{N}{M-1}\rceil$ we are therefore able to provide the bound on $\eta_{1}$.

More specifically, let us start with $j\in\Omega_1$, in which case
\begin{equation}\label{eq:errbnd1}
 \mathbb{\eta}_{j}(t) = \gamma^{-j}\int_0^t e^{ A_{j}^{(1)}(t-s_0)}A_{j+M-1}^{(M)}\mathbf{u}^{\otimes(j+M-1)}\dd s_0 \, .
\end{equation}
Now for purely dissipative problems, as displayed in \cref{eq:R}, we have $\|e^{A_{j}^{(1)}t}\| \leq e^{j \lambda_0t}$.
In addition, $\|A_{j+M-1}^{(M)}\|=j\|F_M\|$ and $\|\mathbf{u}(t)\|\leq\|\mathbf{u}_{\mathrm{in}}\|$, so
\begin{align}\label{eq:errbnd2}
 \|\mathbb{\eta}_{j}(t)\| &\leq j\gamma^{-j}\|F_M\|\|\mathbf{u}_{\mathrm{in}}\|^{j+M-1} \int_0^t e^{j   \lambda_0(t-s_0)}\dd s_0\nn
&= \gamma^{-j} \frac{\|F_M\|\|\uu_{\mathrm{in}}\|^{j+M-1}}{|\lambda_0|} \left( 1-e^{j   \lambda_0t} \right) .
\end{align}
For the case $k=1$, our formula for $f_{j,k,M}$ gives
\begin{equation}
    f_{j,1,M}(\tau) = 1-e^{-j\tau},
\end{equation}
and so we find that for $k=1$ we obtain
\begin{align}
 \|\mathbb{\eta}_{j}(t)\| &\leq \left( \frac{\|\uu_{\mathrm{in}}\|}{\gamma} \right)^j
 \frac{\|F_M\|\|\uu_{\mathrm{in}}\|^{M-1}}{|\lambda_0|} f_{j,1,M}\left( |\lambda_0|t \right) .\label{eq:first_inter}
\end{align}

Now for more general $j\in \Omega_k$, we will show
\begin{equation}
\label{eq:k1_int}
\|\eta_{j}(t)\| \leq \left( \frac{\|\uu_{\mathrm{in}}\|}{\gamma} \right)^j \left(\frac{\|F_M\|\|\mathbf{u}_{\mathrm{in}}\|^{M-1}}{|\lambda_0|}\right)^{k} 
f_{j,k,M}\left( |\lambda_0|t \right)
\end{equation}
We have already shown this for $k=1$ where $j\in\Omega_1$, so we just need to show the iteration step to show the general result by induction.
For $j$ smaller than $N-M+2$, we use \cref{eq:line1}, and \cref{eq:k1_int} for $\eta_{j+M-1}$.
That is,
\begin{equation}
\dt{\eta_{j}} = A_{j}^{(1)}\eta_{j} + \gamma^{M-1}A_{j+M-1}^{(M)}\eta_{j+M-1} \, ,
\end{equation}
leads to
\begin{equation}
 \mathbb{\eta}_{j}(t) = \gamma^{M-1}\int_0^t e^{A_{j}^{(1)}(t-s_1)}A_{j+M-1}^{(M)}\mathbf{\eta}_{j+M-1}(s_1)\dd s_1.
\end{equation}
To upper bound the component $ \mathbb{\eta}_{j}(t)$ with $j\in\Omega_k$, we use the bound in \cref{eq:k1_int} for $j+M-1\in\Omega_{k-1}$.
(That is, we are assuming the expression holds for $k-1$ in order to derive it for $k$.)
That yields the upper bound
\begin{align}
\label{eq:bound_Om2}
 \|\mathbb{\eta}_{j}(t)\| &\leq \gamma^{M-1}\int_0^t \|e^{A_{j}^{(1)}(t-s_1)}\|\cdot \|A_{j+M-1}^{(M)}\mathbf{\eta}_{j+M-1}(s_1)\|\dd s_1 \nonumber\\
 &\leq \gamma^{M-1}j\|F_M\|\int_0^t e^{j\lambda_0(t-s_1)}\|\mathbf{\eta}_{j+M-1}(s_1)\|\dd s_1\nonumber\\
 &\leq \gamma^{M-1}j\|F_M\|
 \left( \frac{\|\uu_{\mathrm{in}}\|}{\gamma} \right)^{j+M-1}\int_0^t e^{j\lambda_0(t-s_1)} 
 \left(\frac{\|F_M\|\|\mathbf{u}_{\mathrm{in}}\|^{M-1}}{|\lambda_0|}\right)^{k-1} 
f_{j+M-1,k-1,M}\left( |\lambda_0|s_1 \right) \nn
&= \frac{\|F_M\| \|\mathbf{u}_{\mathrm{in}}\|^{M-1}}{|\lambda_0|}\left( \frac{\|\uu_{\mathrm{in}}\|}{\gamma} \right)^j
 \left(\frac{\|F_M\| \|\mathbf{u}_{\mathrm{in}}\|^{M-1}}{|\lambda_0|}\right)^{k-1} 
f_{j,k,M}\left( |\lambda_0|t \right) \nn
&= \left( \frac{\|\uu_{\mathrm{in}}\|}{\gamma} \right)^j\left(\frac{\|F_M\| \|\mathbf{u}_{\mathrm{in}}\|^{M-1}}{|\lambda_0|}\right)^{k} 
f_{j,k,M}\left( |\lambda_0|t \right),
\end{align}
where in the second-last line we used Eq.~\eqref{eq:fiter}.
This shows the result \eqref{eq:k1_int} for $k-1$ implies the result for $k$, and therefore for all $k$ by induction.

Since
\begin{equation}
    R = \frac{\|F_M\|}{|\lambda_0|}\|\mathbf{u}_{\mathrm{in}}\|^{M-1} ,
\end{equation}
we obtain the result given in the theorem.
Moreover, since $j=1$ is an element of $\Omega_k$ for $k=\lceil\frac{N}{M-1}\rceil$, we obtain the claimed upper bound on $\|\eta_1\|$.
\end{proof}

\subsection{Bound on the error in terms of max-norm}
\label{app:carl_errMax}
For the case where the ODE is obtained from discretising the ODE, then the above bound on the error for the ODE can be used, but $\|\mathbf{u}_{\mathrm{in}}\|$ will increase with the number of discretisation points.
Instead we can derive a bound using the max-norm, which will be (approximately) independent of the number of discretisation points.
However, we will find that the possibly increasing max-norm with the higher-order discrete Laplacians can cause the bound to be larger than for the first-order discrete Laplacian.

We can start from Eq.~\eqref{eq:errbnd1} above, and use the fact that for the discretised PDE $A_{j+M-1}^{(M)}$ acts only on the components of $\mathbf{u}^{\otimes(j+M-1)}$ with powers of the entries of $\mathbf{u}$.
That means that
\begin{equation}
 \|A_{j+M-1}^{(M)}\mathbf{u}^{\otimes(j+M-1)}\|_{\max} \le
\| A_{j+M-1}^{(M)}\|_{\infty}
\| \mathbf{u}\|_{\max}^{j+M-1} \, .
\end{equation}
Hence we can obtain the equivalent of Eq.~\eqref{eq:errbnd2} as
\begin{equation}
    \|\mathbb{\eta}_{j}(t)\|_{\max}  \lesssim j\gamma^{-j}\|F_M\|_{\infty}\|\mathbf{u}_{\mathrm{in}}\|_{\max}^{j+M-1} \int_0^t 
    \|e^{ A_{j}^{(1)}(t-s_0)}\|_{\infty}
    \dd s_0 \, .
\end{equation}
This expression is obtained in much the same way as before, but using the $\infty$-norm as the induced matrix norm for the max-norm on the vectors.
Here it has been assumed that $\| \mathbf{u}\|_{\max}\le \|\mathbf{u}_{\mathrm{in}}\|_{\max}$, which is correct for the exact PDE, but is not strictly correct for the higher-order discretised PDE.
But, provided the error to the discretisation is appropriately bounded the max-norms should be sufficiently close to that for the continuous PDE, and so $\| \mathbf{u}\|_{\max}\lesssim \|\mathbf{u}_{\mathrm{in}}\|_{\max}$ (which is why we use $\lesssim$ in the expression above).

A further difficulty now is that we need a bound on the $\infty$-norm of $e^{ A_{j}^{(1)}(t-s_0)}$.
For the discretised PDE we consider, we have $F_{1}=DL_{k,d}+c\identity^{\otimes d}$.
First, since $A_{j}^{(1)}$ is constructed from $j$ commuting terms with $F_{1}$ acting on different subsystems, we have
\begin{align}
    \| e^{ A_{j}^{(1)}(t-s_0)} \|_{\infty} &= 
    \| e^{ F_{1}(t-s_0)} \|_{\infty}^j \, .
\end{align}
Next, since the identity commutes with the discretised Laplacian, we have
\begin{align}
    \| e^{ F_{1}(t-s_0)} \|_{\infty}^j &=
    \| e^{ DL_{\kappa,d}(t-s_0)} \|_{\infty}^j e^{jc(t-s_0)} \, .
\end{align}
There we are using $\kappa$ for the order of the discretisation to avoid confusion with $k$ used for iteration in the derivation.
Since the derivatives in the different directions commute, we have
\begin{align}
    \| e^{ DL_{\kappa,d}(t-s_0)} \|_{\infty}^j e^{jc(t-s_0)}
    &= \| e^{ DL_{\kappa}(t-s_0)} \|_{\infty}^{jd} e^{jc(t-s_0)}
    \, .
\end{align}
Ideally we should have that the $\infty$-norm of the discretised Laplacian is equal to 1, but it can be larger by a moderate amount as discussed in Sec.~\ref{sec:stab}.
If we define the maximum of this norm as
\begin{equation}
    G_\kappa := \max_{\tau\ge 0} \| e^{ L_{\kappa}\tau} \|_{\infty} \, ,
\end{equation}
then we have
\begin{align}
    \| e^{ A_{j}^{(1)}(t-s_0)} \|_{\infty} &\le  
    G_\kappa^{jd} e^{jc(t-s_0)} \, .
\end{align}
That then gives us
\begin{align}
    \|\mathbb{\eta}_{j}(t)\|_{\max}  &\lesssim j\gamma^{-j}\|F_M\|_{\infty}\|\mathbf{u}_{\mathrm{in}}\|_{\max}^{j+M-1} G_\kappa^{jd} \int_0^t 
    e^{jc(t-s_0)}
    \dd s_0 \nn
&= \frac{\|F_M\|_{\infty}\|\uu_{\mathrm{in}}\|_{\max}^{M-1}G_\kappa^{jd}}{|c|} \left( 1-e^{j c t} \right) .    
\end{align}

The general form of the bound is then
\begin{equation}
    \|\mathbb{\eta}_{j}(t)\|_{\max} \lesssim\left(\frac{\|F_M\|_{\infty}}{|c|}\|\mathbf{u}_{\mathrm{in}}\|_{\max}^{(M-1)}G_\kappa^{d[j + (k-1) (M - 1)/2]}\right)^{k} 
f_{j,k,M}\left( |c|t \right) .
\end{equation}
This can be shown using iteration as per Eq.~\eqref{eq:bound_Om2} to give
\begin{align}
 \|\mathbb{\eta}_{j}(t)\|_{\max} &\leq \gamma^{M-1}\int_0^t \|e^{A_{j}^{(1)}(t-s_1)}\|_{\infty} \cdot \|A_{j+M-1}^{(M)}\mathbf{\eta}_{j+M-1}(s_1)\|_{\max} \dd s_1 \nonumber\\
 &\leq \gamma^{M-1}j\|F_M\|_{\infty} G_\kappa^{jd} \int_0^t 
    e^{jc(t-s_1)}\|\mathbf{\eta}_{j+M-1}(s_1)\|_{\max} \dd s_1\nonumber\\
 &\lesssim \gamma^{M-1}j\|F_M\|_{\infty} G_\kappa^{jd}
 \int_0^t e^{jc(t-s_1)} 
 \left(\frac{\|F_M\|_{\infty}}{|c|}\|\mathbf{u}_{\mathrm{in}}\|_{\max}^{M-1}G_\kappa^{d[j + k (M - 1)/2]}\right)^{k-1} 
f_{j+M-1,k-1,M}\left( |c|s_1 \right) \nn
&= \frac{\gamma^{M-1}\|F_M\|_{\infty} G_\kappa^{jd}}{|c|}
 \left(\frac{\|F_M\|_{\infty}}{|c|}\|\mathbf{u}_{\mathrm{in}}\|_{\max}^{M-1}G_\kappa^{d[j + k (M - 1)/2]}\right)^{k-1} 
f_{j,k,M}\left( |c|t \right) \nn
&= \left(\frac{\|F_M\|_{\infty}}{|c|}\|\mathbf{u}_{\mathrm{in}}\|_{\max}^{M-1}G_\kappa^{d[j + (k-1) (M - 1)/2]}\right)^{k} 
f_{j,k,M}\left( |c|t \right) .
\end{align}
This upper bound is problematic, because it has $G_\kappa$ to a power increasing with $k$ inside the power of $k$.
This means it is not possible to show convergence with increasing order of the Carleman linearisation.
As the order is increased, $k$ will increase so eventually the expression inside the power of $k$ here will be larger than 1.

A further improvement can be made if we note where the factor of $G_\kappa^{jd}$ is coming from in the second line.
That is coming from $\|e^{A_{j}^{(1)}(t-s_1)}\|_{\infty}$, but  $A_{j}^{(1)}$ is a sum of $j$ operators each acting on a subsystem.
That means $e^{A_{j}^{(1)}(t-s_1)}$ can be written as a tensor product of exponentials of $F_1$ acting on each of the $j$ subsystems.
In turn, $A_{j+M-1}^{(M)}\mathbf{\eta}_{j+M-1}(s_1)$ is a sum of $j$ terms, and each of those $A_{j+M-1}^{(M)}$ has affected only one of the $j$ subsystems.
On the remainder of the subsystems, $\mathbf{\eta}_{j+M-1}(s_1)$ has been obtained by an exponential of the form $e^{A_{j}^{(1)}(s_1-s_2)}$.
That means, if we consider the operators before using the product formulae for norms, we have subsystems where $e^{F_1(t-s_1)}e^{F_1(s_1-s_2)}=e^{F_1(t-s_2)}$.
The product of these exponentials can be bounded by $G_\kappa^{d}$ instead of $G_\kappa^{2d}$.

What this means is that for the $j$ exponentials above, $j-1$ of them can be combined with the previous exponentials, where the norm has already been taken into account as $G_\kappa^{d}$.
That means that the factor need only be $G_\kappa^{d}$ rather than $G_\kappa^{jd}$, and we can obtain the form
\begin{equation}
    \|\mathbb{\eta}_{j}(t)\|_{\max} \lesssim G_\kappa^{dj}\left(\frac{\|F_M\|_{\infty}}{|c|}\|\mathbf{u}_{\mathrm{in}}\|_{\max}^{(M-1)}G_\kappa^{dM}\right)^{k} 
f_{j,k,M}\left( |c|t \right) .
\end{equation}
The iteration to show this is
\begin{align}
 \|\mathbb{\eta}_{j}(t)\|_{\max}
 &\leq \gamma^{M-1}j\|F_M\|_{\infty} G_\kappa^{d} \int_0^t 
    e^{jc(t-s_1)}\|\mathbf{\eta}_{j+M-1}(s_1)\|_{\max} \dd s_1\nonumber\\
 &\lesssim \gamma^{M-1}j\|F_M\|_{\infty} G_\kappa^{d} G_\kappa^{d(j+M-1)}
 \int_0^t e^{jc(t-s_1)} 
 \left(\frac{\|F_M\|_{\infty}}{|c|}\|\mathbf{u}_{\mathrm{in}}\|_{\max}^{M-1}G_\kappa^{dM}\right)^{k-1} 
f_{j+M-1,k-1,M}\left( |c|s_1 \right) \nn
&= \frac{\gamma^{M-1}\|F_M\|_{\infty} G_\kappa^{d(j+M)}}{|c|}
 \left(\frac{\|F_M\|_{\infty}}{|c|}\|\mathbf{u}_{\mathrm{in}}\|_{\max}^{M-1}G_\kappa^{dM}\right)^{k-1} 
f_{j,k,M}\left( |c|t \right) \nn
&= G_\kappa^{dj}\left(\frac{\|F_M\|_{\infty}}{|c|}\|\mathbf{u}_{\mathrm{in}}\|_{\max}^{M-1}G_\kappa^{dM}\right)^{k} 
f_{j,k,M}\left( |c|t \right) .
\end{align}
The advantage of this form is now that the expression inside the exponential is now independent of $k$, and so increasing the order of the Carleman linearisation can improve the accuracy provided the expression inside the exponential is below 1.

\section{Proof of semi-discrete error bound due to finite difference discretisation}\label{app:spat-discretisation-semidiscrete}

The proof to \cref{lem:spat-discretisation-semidiscrete} follows.
\begin{proof}
We want to take account of the error derived from the spatial discretisation of the following rescaled $d$-dimensional nonlinear PDE
\begin{equation}
\label{eq:pde_resc}
    \partial_t u(\xx, t) - \left(D\Delta u(\xx, t) + c u(\xx, t)\right) - b u^M(\xx, t)  = 0.
\end{equation}
When we discretise the equation above, taking $n$ grid points in total, we use the $k$th $d$-dimensional Laplacian approximation $L_{k,d}$, as given in \cref{eq:d_Laplac}, such that
\begin{equation}
L\uu(t) =  L_{k,d}\uu(t) +\mathbf{r}(t),
\end{equation}
where $\mathbf{r}(t)$ is the remainder vector. Using the results given in \cite{kivlichan2017bounding,childs2021high} for the Laplacian approximation, we can write the max-norm of the remainder vector as
\begin{equation}
\label{eq:cons-error}
\|\mathbf{r}\|_{\max}=\norm{ \left(L - L_{k,d}\right)\widetilde{\uu} }_{\max} = \order{   
C(u,k)\left(\frac{e}{2}\right)^{2k} n^{- (2k - 1)/d} },
\end{equation}
where $C(u,k)$ is a constant depending on the $(2k+1)$st spatial derivative on each direction of the Laplacian operator,
\begin{equation}
C(u,k) = \sum_{j=1}^d\left|\frac{d^{2k+1}u}{dx_j^{2k+1}}\right|.    
\end{equation} 
The error in terms of the 2-norm can be obtained immediately from that for the max-norm by multiplying by the factor $\sqrt{n}$
\begin{equation}
\label{eq:consistency-error}
\|\mathbf{r}\|=\norm{ \left(L - L_{k,d}\right)\uu } = \order{   
C(u,k)\sqrt{n}\left(\frac{e}{2}\right)^{2k} n^{-(2k - 1)/d} }.
\end{equation}

In the equation above we also introduce notions of the exact quantities evaluated on the grid points
as $\uu(t) = [u_1, u_2, \dots, u_{n}]^T$ where the components are time dependent. The grid representation of the Laplacian $\Delta$ is denoted $L$, in the sense that $L\uu$ is used to represent the Laplacian of $u$ evaluated at grid points.

To quantify the error due to spatial discretisation we write 
\begin{equation}
\label{eq:approx_disc}
\bm\varepsilon(t) =  \uu(t) -  \uu^h(t),
\end{equation}
where $\uu^h(t)$ are the approximate values of the function resulting from the spatial discretisation. From the definition above we then have
\begin{align}
 \dt{\bm\varepsilon}  &=  D(L\uu - L_{k,d} \uu^h)  +  c(\uu - \uu^h) +  F_{M} (\uu^{\otimes M} - (\uu^h)^{\otimes M})  \nn
  &=  \left(DL_{k,d}+c\identity^{\otimes d}\right)\bm{\varepsilon} +  F_{M}(\uu^{\otimes M} - (\uu^h)^{\otimes M}) +\mathbf{r} \, ,
\end{align}
where we have used \cref{eq:approx_disc} followed by using the operator definitions as given in  \cref{eq:F1}, i.e., $F_{1}=DL_{k,d}+c\identity^{\otimes d}$.
Note that $F_M$ is one-sparse and has only entries equal to $b$ resulting from the spatial discretisation of \cref{eq:pde_resc}.

We can then obtain the derivative of $\|\bm{\varepsilon}\|$ by considering first
\begin{align}
\label{eq:deriv_norm_eps}
    \dt{}(\bm{\varepsilon}^{\dagger}\bm{\varepsilon}) &= \dt{\bm{\varepsilon}^{\dagger}}\bm{\varepsilon} + \bm{\varepsilon}^{\dagger}\dt{\bm{\varepsilon}} \nn
    & = \bm{\varepsilon}^{\dagger}\left( F_{1} +F_{1}^{\dagger}\right)\bm{\varepsilon} + \left(\uu^{\otimes M} - (\uu^{h})^{\otimes M}\right)^{\dagger}F_{M}^{\dagger}\bm{\varepsilon}\nn 
&\quad+  \bm{\varepsilon}^{\dagger}F_{M}\left(\uu^{\otimes M} - (\uu^{h})^{\otimes M}\right) + \mathbf{r}^{\dagger}\bm{\varepsilon} + \bm{\varepsilon}^{\dagger}\mathbf{r} \, .
\end{align}

From the quantity above we have for the first term 
\begin{equation}
\label{eq:first_in}
\bm{\varepsilon}^{\dagger}\left( F_{1} +F_{1}^{\dagger}\right)\bm{\varepsilon} \leq 2 c \|\bm{\varepsilon}\|^2,
\end{equation}
where we have used the results for the eigenvalues of $F_{1}$ from \cref{eq:evalues_F1}, i.e.,
\begin{equation}
\lambda_\ell(F_1) = D n^{2/d} a_0 + c + D n^{2/d} \sum_{j=1}^{k}\left[a_j\left(\omega^{(\ell-1)j} + \omega^{(\ell-1)(n-j)} \right) \right],
\end{equation}
by considering a $d$-dimensional system. The expression for the eigenvalues above $c<0$ is chosen such that it has the maximum real part among all the eigenvalues for $F_1$ and is smaller than zero to guarantee the dissipativity of the PDE. Now for the two other terms, we can first derive the following bound
\begin{align}
   \norm{F_M({\uu}^{\otimes M} - ({\uu}^h)^{\otimes M})}^2 &= \sum_{\ell} b^2 (\uu_\ell^M - ({\uu}_\ell^h)^M)^2
   \nn
&\leq b^2\sum_{\ell}\left(\sum_{j=0}^{M-1}\left(\uu_\ell^h\right)^{j}\left( {\uu}_\ell - {\uu_\ell}^h \right) {\uu_\ell}^{(M-j-1)}\right)^2  \nn 
&\leq b^2\sum_{\ell}\left(\sum_{j=0}^{M-1}\left( {\uu}_\ell - \uu_\ell^h \right) \right)^2 \max\left(\norm{\widetilde{\uu}}_{\max}^{2(M-1)},\norm{\widetilde{\uu}^h}_{\max}^{2(M-1)}\right) \nn 
&\leq b^2M^2\left\|{\bm\varepsilon}\right\|^2\max\left(\norm{\widetilde{\uu}}_{\max}^{2(M-1)},\norm{\widetilde{\uu}^h}_{\max}^{2(M-1)}\right) \, .
\end{align}
Provided $|c|> \uin_{\max}^{M-1}|b|$, the PDE is stable, and we have $\norm{\uu}_{\max} \le \uin_{\max}$.
Then using the approximation that $\norm{\uu}_{\max} \approx \norm{\uu^h}_{\max}$ (the solution is accurate) we obtain
\begin{equation}
    \norm{F_M({\uu}^{\otimes M} - ({\uu}^h)^{\otimes M})} \le bM\left\|{\bm\varepsilon}\right\|\max\left(\norm{{\uu}}_{\max}^{M-1},\norm{{\uu}^h}_{\max}^{M-1}\right)\lesssim bM\left\|{\bm\varepsilon}\right\| {\uin}_{\max}^{M-1} \, .
\end{equation}
Using that expression for the two other terms in \cref{eq:deriv_norm_eps}
\begin{align}
\label{eq:second_in}
\left(\uu^{\otimes M} - (\uu^{h})^{\otimes M}\right)^{\dagger}F_{M}^{\dagger}\bm{\varepsilon} + \bm{\varepsilon}^{\dagger}F_{M}\left(\uu^{\otimes M} - (\uu^{h})^{\otimes M}\right)&\leq 2\|F_M(\uu^{\otimes M} - (\uu^{h})^{\otimes M})\| \|\bm{\varepsilon}\|\nn 
&\le 2 bM\left\|{\bm\varepsilon}\right\|^2 {\uin}_{\max}^{M-1} \, .
\end{align}

For the last two components in \cref{eq:deriv_norm_eps} we have
\begin{equation}
\label{eq:third_in}
\mathbf{r}^{\dagger}\bm{\varepsilon} + \bm{\varepsilon}^{\dagger}\mathbf{r} \leq 2\|\mathbf{r}\|\|\bm{\varepsilon}\|.
\end{equation}
 Returning to \cref{eq:deriv_norm_eps} together with the fact that $\dt{ \|\mathbf{\bm\varepsilon}\|} = (2\|\mathbf{\bm\varepsilon}\|)^{-1} \dt{\|\mathbf{\bm\varepsilon}\|^2}$, we have
\begin{align}
\dt{}\|\bm{\varepsilon\|} &= (2\|\mathbf{\bm\varepsilon}\|)^{-1}\left[\bm{\varepsilon}^{\dagger}\left( F_{1} +F_{1}^{\dagger}\right)\bm{\varepsilon} + \left(\uu^{\otimes M} - (\uu^{h})^{\otimes M}\right)^{\dagger}F_{M}^{\dagger}\bm{\varepsilon}\right.\nn 
&\quad+ \left. \bm{\varepsilon}^{\dagger}F_{M}\left(\uu^{\otimes M} - (\uu^{h})^{\otimes M}\right) + \mathbf{r}^{\dagger}\bm{\varepsilon} + \bm{\varepsilon}^{\dagger}\mathbf{r}\right].
\end{align}
We can now combine \cref{eq:first_in}, \cref{eq:second_in} and \cref{eq:third_in} in the equation above to get
\begin{align}
\dt{}\|\bm{\varepsilon\|}  &\leq \left(c + M|b|{\uin}_{\max}^{M-1}\right)\|\bm{\varepsilon}\| + \|\mathbf{r}\|,
\end{align}
where we used the condition that $|c|>M|b|\uin^{M-1}$.
That is stronger than the condition $|c|> |b|\uin^{M-1}$ which ensures that $\uu(t)$ does not increase with time. We next apply the solution for the ODE as it is given in \cref{eq:pde_sol}, with $\|\bm{\varepsilon}(0)\|=0$ 
\begin{equation}
\|\bm{\varepsilon}(t)\| \leq \frac{1 - \exp{\left( c + M|b|{\uin}_{\max}^{M-1}\right)t} }{ \left|c + M|b|{\uin}_{\max}^{M-1}\right)} \|\mathbf{r}\|,   \end{equation}

We can finally conclude from \cref{eq:consistency-error} that
\begin{equation}
\label{eq:error-decomposition}
\|\bm{\varepsilon}(t)\| = \order{   
C(u,k)\sqrt{n}\left(\frac{e}{2}\right)^{2k} n^{- (2k - 1)/d} \frac{1 - \exp{\left( c + M|b|{\uin}_{\max}^{M-1}\right)t} }{ \left|c + M|b|{\uin}_{\max}^{M-1}\right|}}.
\end{equation}
\end{proof}

\section{Error when using high-order finite difference discretisations}
\label{app:trad_off}

From \cref{eq:cons-error} we have, for $d=1$, the spatial discretisation max-norm error
\begin{equation}
\order{   
C(u,k)\left(\frac{e}{2}\right)^{2k} n^{- (2k - 1)} } .
\end{equation}
Therefore, if we want equivalent error between first-order and order-$k$ discretisations, we would need
\begin{equation}
    C(u,k)\left(\frac{e}{2}\right)^{2k} n_k^{- (2k - 1)}
    \sim
    C(u,1)\left(\frac{e}{2}\right)^{2} n_1^{- 1} .
\end{equation}
That then gives
\begin{equation}
    n_k =\order{  \left( \frac{C(u,k) \, n_1}{C(u,1)} \right)^{1/(2k-1) }}\, .
\end{equation}
This shows that, apart from the factor of $C(u,k)$, increasing the order provides a $(2k-1)$-root improvement in the number of discretisation points needed to achieve a given error.
That is for equal max-norm error.
If we were to instead consider 2-norm error, then there would be a $(4k-3)$-root improvement in the number of points needed.
That is because the lower $n_k$ also reduces the 2-norm by reducing the length of the vector.

In \cref{fig:error_laplacian} we compare the scaling of the error for a first and second-order approximation of the Laplacian (i.e., $k=1,2$ in \cref{eq:ho-laplacian-def}) for a simple differential equation to illustrate the faster error convergence. 
The advantages of using a certain order of approximation $k$ over lower orders will depend on the specific PDE to be solved and the initial conditions.
The value of $C(u,k)$ can increase with $k$ if there is high-frequency variation of $\mathbf{u}$.
That is less of a problem with the Laplace equation, because that smooths out high-frequency variation, but could be more of a problem for other PDEs.
Moreover, the complexity of the block-encoding increases with $k$, so if $k$ is order $n$ then it would negate any quantum speedup in solving the PDE (unless there were a more efficient block-encoding).
Therefore, one should take into consideration the specifics of the PDE to be solved and choose $k$ accordingly.

\begin{figure}
\begin{center}
	\includegraphics[scale=0.6]{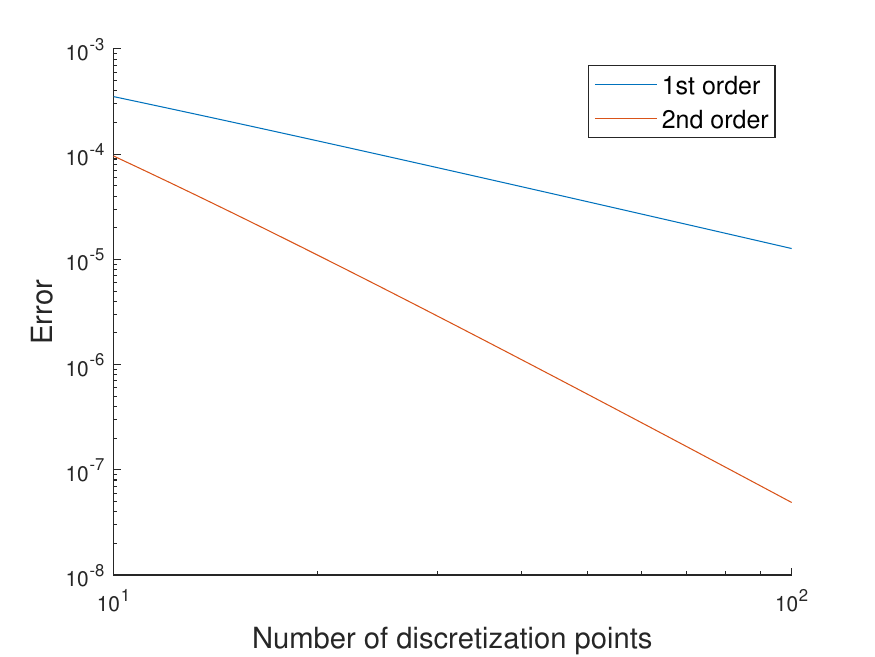}
	\caption{
        Comparison of error in approximating the Laplacian when considering first ($k=1$) and second ($k=2)$ order discretisations. To illustrate the scaling of the error we have solved the equation $u_{xx}=e^x$ in the interval $x\in[0,1]$ with Dirichlet boundary conditions $u(0)=0$ and $u(1)=1$. The error is computed as $\norm{L_k^{-1}\exp[\mathbf{x}]-\mathbf{u}_{\text{sol}}}$, where $\mathbf{x}$ is the vector of discrete points in space, $\exp[\mathbf{x}]$ is obtained by exponentiating each entry of $\mathbf{x}$ and $\mathbf{u}_{\text{sol}}$ is the vector containing the exact discretised solution evaluated at the grid points.}\label{fig:error_laplacian}
\end{center}
\end{figure}

\section{Construction of finite-difference coefficients}\label{app:fd-coefficients}
For the construction of arbitrary finite-difference coefficients, we generally refer to standard methods, such as based on Taylor expansions or Lagrange polynomials \cite{fdcc2016}. 
Given a grid discretisation of $x$ over some finite interval in $n$ points, a finite difference approximation or ``stencil'' to the $m$th derivative of a function $f(x)$ reads
\begin{equation}
    f^{(m)}(x) \approx \sum_{j=0}^{n-1} a_j f(x_j),
\end{equation}
where $\{a_j\}_{j=0}^{n-1}$ are the stencil points to the grid $\{x_j\}_{j=0}^{n-1}$. When finding these approximations, one may  express $f(x)$ via its Taylor expansion or as a Lagrange polynomial to find the coefficients $\{a_j\}$ given a certain grid. Generally, the number of grid points used in approximating a certain derivative will influence its accuracy; we refer to Ref.~\cite{li2005general} which provides a comprehensive analysis.

Within this study, we are mostly interested in the discretisation of second derivatives, where we focus on central stencils with periodic boundary conditions. To that end, one can make use of the formulae in Ref.~\cite{li2005general}, which as well as the stencil coefficients also provides an investigation of the discretisation error.
Whenever we talk about ``order of discretisation'' $k$ in this work, we have a central stencil with $2k+1$ points, corresponding to an error that goes roughly as $\order{h^{2k-1}}$ in a uniform grid spacing $h$.

The formulae in Ref.~\cite{li2005general} are also considered in other works that apply finite differences to quantum algorithms~\cite{kivlichan2017bounding,childs2021high}.
This way, an approximation to the second derivative for a function $f(x)$ can be written as
\begin{align} \label{eq:ho-stencil}
    f''(0) &\approx \frac{1}{h^2} \sum_{j=-k}^k a_j f(jh), \quad \text{where} \\ 
    a_j &= \begin{cases}
        \frac{2(-1)^{j+1}(k!)^2}{j^2 (k-j)! (k+j)!}, &\quad 1 \le j \le k \\
        -2\sum_{i=1}^k a_i , &\quad  j=0 \\ 
        a_{-j}, &\quad -k \ge j \ge -1.
    \end{cases} 
\end{align}
For $k=1$
the usual stencil with coefficients $1, -2, 1$ is retrieved.

Using this construction, the discrete Laplacian matrix $L_k$ with a uniform grid spacing $h$ has entries
\begin{equation}
    [L_{k}]_{pq} = \begin{cases}
        \frac{2(-1)^{\ell+1}(k!)^2}{\ell^2 (k-\ell)! (k+\ell)!}, &\quad q<p\le q+k~~\text{or}\\ &~~q-k\le p < q, ~~ \ell:=\abs{p-q} \\
        -2\sum_{\ell=q+1}^{q+k} [L_{k}]_{\ell q} , &\quad  p=q 
    \end{cases}
\end{equation}
To treat boundaries, one can use the method in Ref.~\cite{svard2006order}, which showns that if using a $k$th order finite difference scheme within the domain, it is sufficient to treat the boundary with a stencil of accuracy $k-2$, given that the solution is point-wise bounded. We need to make sure that only internal points are used for the creation of the difference stencils; thus, for rows with row-index smaller than $k$ or greater than $n-k$, we need to construct specific, non-centric stencils.
The procedure for boundary treatment in Ref.~\cite{costa2019quantum} may be used as well.

\section{Sparsity of Carleman matrix}\label{app:sparsity-of-carlmatrix}
Here we discuss the sparsity of the Carleman matrix applied to \cref{eq:pde-example}.
\begin{lemma}[Sparsity of the Carleman matrix for nonlinear PDE problem]
Let $\carlmatrix$ be the Carleman matrix with truncation number $N$ applied to discretised PDE problem $\partial_t u = D L_k u + cu + bu^M$ in \cref{eq:pde-example}, where $L_k$ is a central difference approximation of the Laplacian with $2k+1$ stencil points under periodic boundary conditions.  Then, considering only one dimension in space, the sparsity $s(\carlmatrix)$ of $\carlmatrix$ is 
\begin{equation}
    s(\carlmatrix) = \order{kN}.
\end{equation}
\end{lemma}

\begin{proof}
As individual diagonal blocks in $\mathcal{A}_N$ (see Eq.~\eqref{eq:carlmatrix}) do not overlap, we can treat them independently and sum up the contributions.
    We first consider the linear term $F_1=DL_k+c\identity$. As the diagonal entries of $L_k$ are non-zero, we can equivalently only consider $L_k$.
    The central-difference stencil is of the form
     $\frac{1}{h^2} \sum_{\ell=-k}^k r_\ell f(\ell h)$,
    hence for each row of $L_k$, there are $2k+1$ non-zero contributions, where there is one diagonal entry and $2k$ off the diagonal.
    The diagonal blocks of $\carlmatrix$ are of the form
    \begin{equation}
        A^{(1)}_{j} = F_1\otimes \identity^{\otimes (j-1)} + \identity\otimes F_1\otimes \identity^{\otimes(j-2)} + \cdots +\identity^{\otimes (j-1)}\otimes F_1 \, .
    \end{equation}
    Each term is of the form of the tensor product of $F_1$ with identity matrices, which leaves the sparsity unchanged.
    The sparsity of $A^{(1)}_{j}$ is then upper bounded by the sum of the sparsities of the terms, and so is upper bounded by $j(2k+1)$.
    Since the maximum value of $j$ is $N$, that gives an upper bound of $N(2k+1)$.
    In fact, since each term has one on-diagonal and $2k$ off-diagonal entries in each row, the sparsity will be $2kN+1$.
    
Next, consider the off-diagonal blocks of $\carlmatrix$, given by (see \cref{eq:carleman-assembly})
\begin{align}
A^{(M)}_{j+M-1} = F_M\otimes \identity^{\otimes (j-1)} + \identity\otimes F_M\otimes \identity^{(j-2)} + \cdots+ \identity^{\otimes (j-1)}\otimes F_M \, .
\end{align}    
The $n\times n^M$ matrices $F_M$ have the only non-zero entries (see \cref{eq:FM-definition} for the case $M=2$)
\begin{align}
    (F_M)_{i, i+(i-1)(\sum_{\ell=1}^{M-1}n^\ell)} = b \, .
\end{align}    
Hence $F_M$ is one-sparse, and again taking the tensor product with the identity leaves the sparsity unchanged.
The sparsity of $A^{(M)}_{j+M-1}$ is then no more than the sum of the sparsities of the terms, and so is no larger than $j$.
The last ($N$th) column of $\carlmatrix$ has $A_{j+M-1}^{(M)}$ with $j+M-1=N$, and so the maximum value of $j$ is $N-M+1<N$.

The overall sparsity of $\carlmatrix$ is then upper bounded by the sum of the sparsities of $A^{(1)}_{j}$ (on-diagonal blocks) and $A^{(M)}_{j+M-1}$ (off-diagonal blocks).
That gives the total sparsity upper bounded as
\begin{equation}
    s(\carlmatrix) < N(2k+1) + N \in\order{kN} \, .
\end{equation}
\end{proof}

\section{Block-encoding of Carleman matrix $\mathcal{A}_N$}\label{app:block_enc}

As usual in the block encoding of an operator written as a sum, the sum can be block encoded by using a register in superposition to select between the terms in the sum, and then a select operation controlled by that register to implement the terms in the sum.
Here we have a sum of block diagonals with $A_j^{(1)}$ and $A_{j+M-1}^{(M)}$.
We therefore have the following approach to the block encoding.
\begin{enumerate}
    \item A qubit in superposition is used to select between $A_j^{(1)}$ and $A_{j+M-1}^{(M)}$.
    \item We prepare ancilla registers in equal superpositions over $N$ and $N-M+1$ basis states; these register store $\ell$.
    These registers are for the block encodings of $A_j^{(1)}$ and $A_{j+M-1}^{(M)}$, respectively.
    \item For the block encoding of $A_j^{(1)}$ we apply $F_1$ to target register $\ell$.
    Instead of using $N$ controlled block encodings of $F_1$, we can perform a controlled swap of register $\ell$ to a working register, use \emph{one} controlled block encoding of $F_1$, then swap the result back.
    These controlled swaps have cost $\order{N\log n}$ for dimension $n$ of the vector, in terms of elementary gates.
    \item For the block encoding of $A_{j+M-1}^{(M)}$ we perform a controlled application of $F_M$ on target registers $\ell$ to $\ell+M-1$ (with $\ell$ starting from 1).
    Rather than using $N-M+1$ controlled block encodings of $F_M$, we can swap the $M$ target registers to working registers, then perform one controlled block encoding of $F_M$, then swap the result back.
    These controlled swaps have cost $\order{NM\log n}$.
    \item 
    For block encoding $A_{j+M-1}^{(M)}$ a further complication is that $F_M$ maps $M$ copies to 1, so we need to shift the target systems over so we have a contiguous block of copies.
    That is, if $\ell=1$, we have the first term in the sum for $A_{j+M-1}^{(M)}$, the first $M$ copies are mapped to 1 and we have the remaining $M-1$ copies being zeroed, and then there are the rest of the copies.
    The $M-1$ zeroed copies need to be shifted to the end.
    These controlled swaps have a further cost $\order{NM\log n}$ in elementary gates.
    We also subtract $M-1$ from $m$, to indicate that the number of copies has been reduced by $M-1$.
    That corresponds to the blocks containing $A_{j+M-1}^{(M)}$ being away from the main diagonal.
    \item 
    Lastly, we need to truncate the sums for $A_j^{(1)}$ and $A_{j+M-1}^{(M)}$ to $j$ terms.
    This can be performed by using an inequality test between $\ell$ and $m$, with a result of $1$ removing the term in the block encoding.
    In particular, for $A_j^{(1)}$ we have $j$ corresponding to $m$, and $\ell>m$ indicates that we are past the end of the sum.
    Then, for $A_{j+M-1}^{(M)}$ we have $j+M-1=m$, so $\ell>m-M+1$ indicates we are past the end of the sum.
\end{enumerate}

It is also possible to prepare the initial Carleman vector with $\order{N}$ calls to the state preparation for $\mathbf{u}_{\mathrm{in}}$.
We may first prepare a register in unary for the appropriate weightings between the components of the Carleman vector.
Then simply use each qubit of that unary register to control preparation of the $\mathbf{u}_{\mathrm{in}}$ state.
Then we obtain a number of copies corresponding to the component of the Carleman vector as required.

\end{document}